\documentclass[a4paper,USenglish]{lipics-v2021}
\usepackage[vlined,noend,linesnumbered,noresetcount]{algorithm2e}

\usepackage{cite}

\nolinenumbers
\hideLIPIcs

\newif\iflong
\longtrue %

\title{Fault-Tolerant Computing with Unreliable Channels (Extended Version)}

\author{Alejandro Naser-Pastoriza}
       {IMDEA Software Institute, Madrid, Spain \and
        Universidad Politécnica de Madrid, Spain}{}{}{}
\author{Gregory Chockler}
       {University of Surrey, Guildford, UK}{}{}{}
\author{Alexey Gotsman}
       {IMDEA Software Institute, Madrid, Spain}{}{}{}

\authorrunning{A. Naser-Pastoriza, G. Chockler and A. Gotsman}
\titlerunning{Fault-Tolerant Computing with Unreliable Channels}


\funding{This work was partially supported by the PRODIGY and DECO projects
  funded by MCIN/AEI, the BLOQUES project funded by the Madrid regional
  government, and by a research grant from Nomadic Labs.}

\Copyright{A. Naser-Pastoriza, G. Chockler and A. Gotsman}
\EventEditors{Alysson Bessani, Xavier D\'efago, Yukiko Yamauchi, and Junya Nakamura}
\EventNoEds{4}
\EventLongTitle{27th International Conference on Principles of Distributed Systems (OPODIS 2023)}
\EventShortTitle{OPODIS 2023}
\EventAcronym{OPODIS}
\EventYear{2023}
\EventDate{December 6--8, 2023}
\EventLocation{Tokyo, Japan}
\EventLogo{}
\SeriesVolume{286}
\ArticleNo{21}

\begin{document}

\newtheorem*{theorem*}{Theorem}

\newtheorem{innercustomthm}{Theorem}
\newenvironment{customthm}[1]
  {\renewcommand\theinnercustomthm{#1}\innercustomthm}
  {\endinnercustomthm}

\newtheorem{innercustomcor}{Corollary}
\newenvironment{customcor}[1]
  {\renewcommand\theinnercustomcor{#1}\innercustomcor}
  {\endinnercustomcor}

\theoremstyle{claimstyle}
\newtheorem{myclaim}{Claim}

\SetKwProg{Function}{function}{:}{}
\SetKwBlock{SubAlgoBlock}{}{end}
\newcommand{\SubAlgo}[2]{#1 \SubAlgoBlock{#2}}
\let\oldnl\nl
\newcommand{\nonl}{\renewcommand{\nl}{\let\nl\oldnl}}

\newcommand{\tr}[2]{\iflong{}\S#1\else{}\cite[\S#2]{this-extended-version}\fi}
\newcommand{\appabd}{A}
\newcommand{\appregfromconsensus}{B}
\newcommand{\appsync}{C}
\newcommand{\appconsensus}{D}

\newcommand{\when}{\textbf{when}\xspace}
\newcommand{\onreceive}{\textbf{when received}\xspace}
\newcommand{\onaccumulate}{\textbf{when accumulated}\xspace}
\newcommand{\send}{\textbf{send}\xspace}
\newcommand{\from}{\textbf{from}\xspace}
\newcommand{\ToAll}{\textbf{to all}\xspace}
\newcommand{\To}{\textbf{to}\xspace}
\newcommand{\precond}{\textbf{pre:}\xspace}
\newcommand{\poscond}{\textbf{pos:}\xspace}
\newcommand{\assign}[2]{\ensuremath{#1} \ensuremath{\leftarrow} \ensuremath{#2}}
\newcommand{\trigger}{\textbf{trigger}\xspace}
\newcommand{\nop}{{\bf nop}\xspace}
\newcommand{\whentimer}{{\bf when the timer }\xspace}
\newcommand{\delete}{\textbf{delete}\xspace}
\newcommand{\store}{\textbf{store}\xspace}

\newcommand{\FromAQuorum}{\from \ {\bf each $p_j$ in a quorum} $\mathcal{Q}$}
\newcommand{\ToAllMinusSelf}{\To $\mathcal{P} \setminus \{p_i\}$}
\newcommand{\AQuorumOf}{\ {\bf a quorum of}\ }

\newcommand{\view}{\mathsf{view}}
\newcommand{\currview}{\mathsf{curr\_view}}
\newcommand{\views}{\mathsf{views}}
\newcommand{\viewtimer}{\mathsf{view\_timer}}
\newcommand{\expired}{\mathsf{expired}}
\newcommand{\TRUE}{{\tt TRUE}}
\newcommand{\FALSE}{{\tt FALSE}}
\newcommand{\advanced}{\mathsf{advanced}}
\newcommand{\status}{\mathsf{status}}
\newcommand{\phase}{\mathsf{phase}}
\newcommand{\cview}{\mathsf{cview}}
\newcommand{\logg}{\mathsf{log}} %
\newcommand{\nextt}{{\sf next}} %
\newcommand{\lastdelivered}{{\sf last\_delivered}}

\newcommand{\onexpire}{\textbf{when the timer}\xspace}
\newcommand{\expires}{\textbf{expires}\xspace}
\newcommand{\enterview}{\mathsf{enter\_view}}
\newcommand{\newview}{\mathsf{new\_view}}
\newcommand{\stoptimer}{\mathsf{stop\_timer}}
\newcommand{\starttimer}{\mathsf{start\_timer}}
\newcommand{\stopalltimers}{\mathsf{stop\_all\_timers}}
\newcommand{\adv}{\mathsf{advance}}
\newcommand{\start}{\mathsf{start}}
\newcommand{\timerdelivery}{\mathsf{timer\_delivery}}
\newcommand{\timerrecovery}{\mathsf{timer\_recovery}}
\newcommand{\timercommit}{\mathsf{timer\_commit}}
\newcommand{\durdelivery}{\mathsf{dur\_delivery}}
\newcommand{\durrecovery}{\mathsf{dur\_recovery}}
\newcommand{\durcommit}{\mathsf{dur\_commit}}
\newcommand{\broadcast}{{\sf broadcast}}
\newcommand{\deliver}{{\sf deliver}}

\newcommand{\GST}{{\sf GST}}
\newcommand{\GV}{{\sf GV}}
\newcommand{\LV}{{\sf LV}}
\newcommand{\dom}{{\sf dom}}
\newcommand{\F}{\mathcal{F}}
\newcommand{\PP}{\mathcal{P}}
\newcommand{\GG}{\mathcal{G}}
\newcommand{\V}{\mathcal{V}}
\newcommand{\Q}{\mathcal{Q}}
\newcommand{\Sy}{\mathcal{S}}
\newcommand{\hub}{\mathcal{H}}
\newcommand{\E}{\overline{E}}
\newcommand{\A}{\overline{A}}
\newcommand{\leader}{{\sf leader}}

\newcommand{\WISH}{{\tt WISH}}
\newcommand{\ENTER}{{\tt ENTER}}
\newcommand{\STATE}{{\tt STATE}}
\newcommand{\NEWSTATE}{{\tt NEW\_STATE}}
\newcommand{\NEWSTATEACK}{{\tt NEW\_STATE\_ACK}}
\newcommand{\BROADCAST}{{\tt BROADCAST}}
\newcommand{\ACCEPT}{{\tt ACCEPT}}
\newcommand{\ACCEPTACK}{{\tt ACCEPT\_ACK}}
\newcommand{\COMMIT}{{\tt COMMIT}}
\newcommand{\RESENDCOMMIT}{{\tt RESEND\_COMMIT}}
\newcommand{\BROADCASTSTATE}{{\tt BROADCAST\_STATE}}
\newcommand{\ABORT}{{\tt ABORT}}

\newcommand{\OGST}{{\sf GST} + \rho}
\newcommand{\AF}[1]{A_{\rm first}^{#1}}
\newcommand{\AL}[1]{A_{\rm last}^{#1}}
\newcommand{\EF}[1]{E_{\rm first}^{#1}}
\newcommand{\EL}[1]{E_{\rm last}^{#1}}
\newcommand{\AEF}{{E}_{\rm first}}
\newcommand{\AEL}{{E}_{\rm last}}
\newcommand{\AAF}{{A}_{\rm first}}
\newcommand{\AAL}{{A}_{\rm last}}
\newcommand{\CF}{C_{\rm fin}}
\newcommand{\CI}{C_{\rm inf}}
\newcommand{\fdef}{\mathpunct{\downarrow}}
\newcommand{\fndef}{\mathpunct{\uparrow}}

\newcommand{\RECOVERING}{\textsc{recovering}}
\newcommand{\FOLLOWER}{\textsc{follower}}
\newcommand{\LEADER}{\textsc{leader}}
\newcommand{\ADVANCED}{\textsc{advanced}}
\newcommand{\ELECTION}{\textsc{election}}

\renewcommand{\_}{\texttt{\textunderscore}}

\newcommand{\diameter}{\mathsf{diameter}}

\newcommand{\propose}{\mathsf{propose}}
\newcommand{\decide}{\mathsf{decide}}
\newcommand{\val}{\mathsf{val}}
\newcommand{\pval}{\mathsf{my\_proposal}}
\newcommand{\POA}{{\tt P1A}}
\newcommand{\POB}{{\tt P1B}}
\newcommand{\PTA}{{\tt P2A}}
\newcommand{\PTB}{{\tt P2B}}
\newcommand{\MSGOB}{\mathsf{M1B}}
\newcommand{\MSGTA}{\mathsf{M2A}}
\newcommand{\MSGTB}{\mathsf{M2B}}
\newcommand{\VOB}{\textit{V1B}}
\newcommand{\VTA}{\textit{V2A}}
\newcommand{\VTB}{\textit{V2B}}

\newcommand{\timerdecision}{\mathsf{decision\_timer}}
\newcommand{\durdecision}{\mathsf{timeout}}

\newcommand{\RG}[2]{{#1}\setminus{#2}}
\newcommand{\FS}{\mathcal{F}}
\newcommand{\KFS}{\mathcal{F}_M}
\newcommand{\CC}{\mathcal{C}}
\newcommand{\CQ}{\mathcal{S}}
\newcommand{\ALG}{\mathcal{A}}
\newcommand{\RD}{{\it read }}
\newcommand{\MRL}{\mathcal{M}_{{\rm ARD}}}
\newcommand{\MRU}{\mathcal{M}_{{\rm AEF}}}
\newcommand{\MCL}{\mathcal{M}_{{\rm PRD}}}
\newcommand{\MCU}{\mathcal{M}_{{\rm PEF}}}

\newcommand{\M}{\mathcal{M}}

\newcommand{\MRUH}{\mathcal{S}^{\star}}

\newcommand{\itval}{{\it val}}
\newcommand{\itview}{{\it view}}
\newcommand{\itcview}{{\it cview}}

\newcommand{\ENTERED}{{\tt ENTERED}}
\newcommand{\PROPOSED}{{\tt PROPOSED}}
\newcommand{\ACCEPTED}{{\tt ACCEPTED}}
\newcommand{\DECIDED}{{\tt DECIDED}}

\newcommand{\return}{\textbf{return}\xspace}
\newcommand{\wait}{\textbf{wait until}\xspace}

\newcommand{\idle}{\textsc{idle}}
\newcommand{\wrquery}{\textsc{wr\_query}}
\newcommand{\wrpropagate}{\textsc{wr\_propagate}}
\newcommand{\wrdone}{\textsc{wr\_done}}
\newcommand{\rdquery}{\textsc{rd\_query}}
\newcommand{\rdpropagate}{\textsc{rd\_propagate}}
\newcommand{\rddone}{\textsc{rd\_done}}

\newcommand{\oread}{\mathsf{read}}
\newcommand{\owrite}{\mathsf{write}}
\newcommand{\ts}{\mathsf{tag}}
\newcommand{\wrval}{\mathsf{wr\_val}}
\newcommand{\rdval}{\mathsf{rd\_val}}
\newcommand{\seq}{\mathsf{seq}}
\newcommand{\query}{\mathsf{query}}
\newcommand{\queryack}{\mathsf{query\_ack}}
\newcommand{\wwrite}{\mathsf{write}}
\newcommand{\wwriteack}{\mathsf{write\_ack}}
\newcommand{\retval}{\textit{ret\_val}}

\newcommand{\sseq}{\mathit{seq}}
\newcommand{\sval}{\mathit{val}}
\newcommand{\sts}{\mathit{tag}}

\newcommand{\ACK}{{\tt ACK}}
\newcommand{\QUERY}{{\tt QUERY}}
\newcommand{\QUERYACK}{{\tt QUERY\_ACK}}
\newcommand{\WRITE}{{\tt WRITE}}
\newcommand{\WRITEACK}{{\tt WRITE\_ACK}}

\newcommand{\requested}{\mathsf{requested}}
\newcommand{\pending}{\mathsf{pending}}
\newcommand{\ordered}{\mathsf{ordered}}
\newcommand{\completed}{\mathsf{completed}}
\newcommand{\responded}{\mathsf{responded}}
\newcommand{\vertices}{\mathsf{vertices}}
\newcommand{\lastwrites}{\mathsf{lastwrites}}
\newcommand{\precc}{\mathsf{prec}}
\newcommand{\tagg}{\mathsf{tag}}
\newcommand{\mintag}{\mathsf{mintag}}
\newcommand{\dagg}{\mathsf{dag}}
\newcommand{\edges}{\mathsf{edges}}
\newcommand{\ops}{\mathcal{O}}
\newcommand{\newtag}{\mathsf{newtag}}

\maketitle
\begin{abstract} 
  We study implementations of basic fault-tolerant primitives, such as consensus
  and registers, in message-passing systems subject to process crashes and a
  broad range of communication failures. Our results characterize the necessary
  and sufficient conditions for implementing these primitives as a function of
  the connectivity constraints and synchrony assumptions. Our main contribution
  is a new algorithm for partially synchronous consensus that is resilient to
  process crashes and channel failures and is optimal in its connectivity
  requirements. In contrast to prior work, our algorithm assumes the most
  general model of message loss where faulty channels are {\em flaky}, i.e., can
  lose messages without any guarantee of fairness. This failure model is
  particularly challenging for consensus algorithms, as it rules out standard
  solutions based on leader oracles and failure detectors. To circumvent this
  limitation, we construct our solution using a new variant of the recently
  proposed \emph{view synchronizer} abstraction, which we adapt to the
  crash-prone setting with flaky channels.

\end{abstract}

\bigskip
\section{Introduction}
\label{sec:intro}

We are concerned with implementing basic fault-tolerant primitives, such as
registers and consensus, in systems where processes can crash and communication
channels may lose messages.  This setting is of practical importance, since
channel failures regularly occur in real-world
deployments~\cite{bailis-kingsbury,physalia,heidi-raft}. They arise from a
variety of reasons -- physical infrastructure failures, bugs in switch software,
configuration errors -- and they often lead to system outages. For example,
Alquraan et al.~\cite{osdi-partitions} conducted a comprehensive study of
failures due to network partitions in widely used replicated data stores.  It
found that a majority of such failures led to catastrophic effects and that the
resolution of almost half of these failures required redesigning a system
mechanism. Thus, the failures were not simply due to coding bugs, but to design
flaws.

\begin{figure}[t]
  \centering
    \includegraphics[width=0.85\linewidth,trim=0 240 0 240,clip]{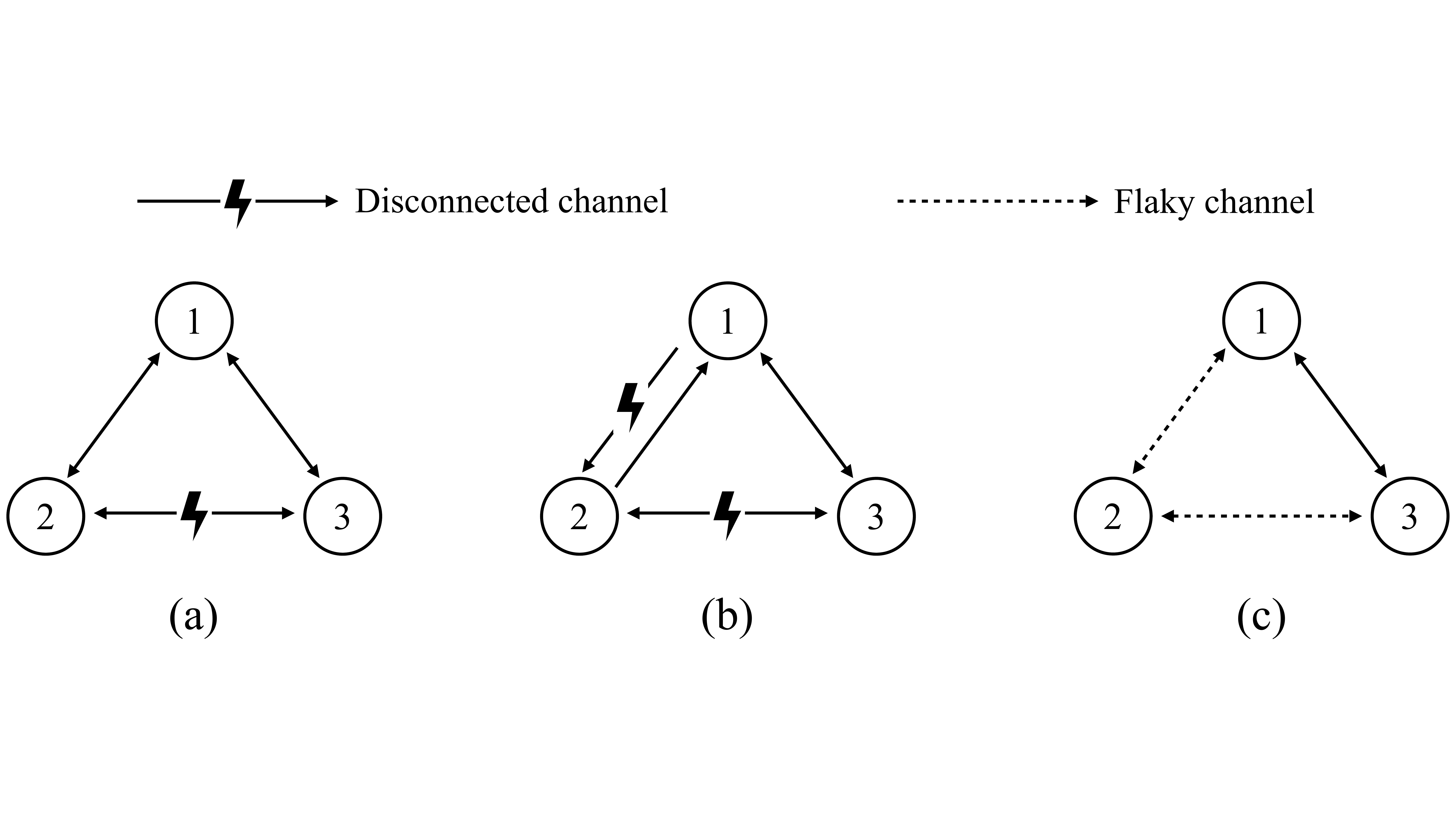}
    \caption{Examples of irregular connectivity configurations:
    (a) indirect connectivity; (b) asymmetric connectivity; and (c) flaky
    connectivity. All processes are correct.}
    \label{fig:examples}
\end{figure}

Perhaps the most challenging aspect of real-world network failures is that
the ways the network can break down can be arbitrarily complex,
resulting in a wide variety of connectivity configurations.
While in the simplest case individual channel failures may not affect the overall
network connectivity (Figure~\ref{fig:examples}a), in more complex scenarios
the network may become partitioned into multiple components,
some of which may end up connected in one direction,
but not in the other (Figure~\ref{fig:examples}b).
In the worst case, some components may become intermittently connected,
causing an arbitrary subset of 
the messages transmitted between them to get lost (Figure~\ref{fig:examples}c).

Intermittent connectivity has so far received little attention in the theory
research. Most network models for studying consensus assume that channels are
either (eventually) reliable or (eventually) disconnected
(Figures~\ref{fig:examples}a-b). Reliable channels are sometimes replaced by
\emph{fair lossy} ones, which only guarantee to deliver a message if it was sent
infinitely many times. However, reliable and fair lossy channels are
computationally equivalent: the former can be implemented from the latter by
repeatedly resending each message until it gets acknowledged and filtering out
duplicates~\cite{lossy-channels-AfekAFFLMWZ94}.

Although the classical abstractions of failure or leader detectors~\cite{CT96}
can be adapted to solve consensus in the above
settings~\cite{paxos,raft,aguilera-heartbeat,friedman1,friedman2,friemdan-podc-ba},
they are no longer useful in the presence of unconstrained message loss.
Intuitively, the reason is that in this case they may fail to correctly identify
processes with reliable connectivity, which is necessary to ensure the liveness
of consensus.
For example, suppose that the pattern of message loss experienced by the channels $2 \leftrightarrow 1$ and
$2 \leftrightarrow 3$ in Figure~\ref{fig:examples}c is such that all leader
election messages are getting through. Thus, it is possible for process $2$
to be elected as a leader. However, since any message 
sent by process $2$ after being elected can be dropped (no matter how many times
it is resent),
the process may not be able to drive consensus to completion, violating liveness.

In this paper 
we investigate the possibility of implementing
basic fault-tolerant abstractions resilient to a broad range of communication
failures, including those induced by intermittent connectivity.
To this end, we obtain several lower and upper bounds that
characterize the solvability of registers and consensus
as a function of the channel failure model, 
connectivity constraints and synchrony assumptions.
Our main contribution is a new crash fault-tolerant algorithm that solves
partially synchronous consensus under the \emph{most general} model
of message loss and is \emph{optimal}
in terms of its network connectivity requirements. Our
algorithm furthermore offers a new way of constructing modular protocols 
for unreliable networks, which does not rely on the leader election or
failure detection abstractions. 
Below we review our results in more detail.

\subparagraph{Lower bounds.}  We first present a general framework for
specifying fault-tolerance assumptions on process and channel failures
(\S\ref{sec:model}-\ref{sec:correctness}) by generalizing the classical notion
of a {\em fail-prone system}~\cite{bqs}.
For example, this framework allows us to specify whether an algorithm is meant
to tolerate the failure patterns depicted in Figure~\ref{fig:examples}.
We then use this framework to establish minimal
connectivity requirements necessary to implement registers or consensus
(\S\ref{section:lower_bound}). To obtain strong lower bounds, we consider strong
network models where correct channels are reliable and faulty channels are disconnected,
i.e., drop all messages. We consider asynchrony for registers (model $\MRL$ in
Figure~\ref{fig:models}) and partial synchrony~\cite{dls} for consensus (model
$\MCL$): the execution starts with an asynchronous period and then becomes
synchronous. The most interesting aspect of our lower bounds is that, unlike 
the prior work on consensus under unreliable connectivity
(e.g.,~\cite{aguilera-podc04,topology-consensus}), they are proven for
a very weak termination guarantee that only requires 
obstruction-free termination at a \emph{subset} of processes. 
Informally, we show that for any $n$-process implementation of a register or consensus:
\begin{enumerate}
\item All processes where obstruction-freedom holds must be strongly connected
  via correct channels.
\item If the implementation tolerates $k$ process crashes and $n=2k+1$, then any
  process where obstruction-freedom holds must belong to a set of $\ge k+1$
  correct processes strongly connected by correct channels.
\end{enumerate}
We call the largest set satisfying the condition in (2) the {\em connected core}
of the network. For example, in Figure~\ref{fig:examples}a the connected core is
$\{1, 2, 3\}$, whereas in Figures~\ref{fig:examples}b-c it is $\{1, 3\}$.

The above results generalize the celebrated CAP theorem~\cite{brewer,cap}, which
says that it is impossible to guarantee all of Consistency, Availability and
network Partition-tolerance (\S\ref{sec:cap-revisited}). Whereas CAP only says
that to ensure availability and consistency, {\em some} channels must be
correct, we establish {\em how many} are needed; and whereas CAP requires
availability at {\em all} processes, we establish conditions required to ensure
it only in {\em a part} of the system. Result (2) furthermore establishes
stronger requirements for algorithms resilient to a minority of crashes (which
is optimal~\cite{lynch-dist-algos,dls}).
It shows that, even if obstruction-freedom is required only at a {\em single}
process, a majority thereof must still be strongly connected by correct
channels.

\begin{figure}[t]
\begin{center}
\begin{tabular}{|c|c|c|}
  \hline
  & Reliable/Disconnected
  &
    Eventually reliable/Flaky \\
  & (lower bound)
  &
    (upper bound) \\
  \hline
  Asynchronous (register)&
  $\MRL$ &
  $\MRU$
\\
  \hline
  Partially synchronous (consensus) &
  $\MCL$ &
  $\MCU$
\\
  \hline
\end{tabular}
\end{center}
\caption{System models with different channel types and synchrony assumptions.}
\label{fig:models}
\end{figure}

\subparagraph{Upper bounds.}  The second part of our contribution is to propose
algorithms for registers and consensus that match our lower bounds. These
algorithms assume the existence of the connected core and guarantee wait-freedom
at all of its members.
They are designed to work in an adversarial model where correct channels are only
{\em eventually} reliable and faulty channels are {\em flaky}, i.e., can drop
any messages sent on them without any guarantee of fairness
(models $\MRU$ and $\MCU$ in Figure~\ref{fig:models}). Flakiness is a very
broad failure mode that captures a number of failure patterns occurring in
practice, including both full and intermittent loss of connectivity. It also
subsumes the previously considered variants of lossy channels, such as
eventually disconnected~\cite{aguilera-heartbeat,gcs-survey} and fair
lossy~\cite{lossy-channels-AfekAFFLMWZ94,basu96}. Although irregular connectivity
patterns have been studied before~\cite{friedman1,friedman2,friemdan-podc-ba,aguilera-heartbeat,aguilera-podc03-journal},
to the best of our knowledge we are the first to propose register and consensus
implementations in the presence of flaky channels.

The implementation for consensus is the more challenging one
(\S\ref{section:consensus}). Since, as we explained above, failure and leader
detectors are not useful in the presence of flaky channels, we take a different
approach: we generalize the abstraction of a {\em view
  synchronizer}~\cite{multishot-disc22,bftlive-dc,cogsworth,oded-linear},
recently proposed for Byzantine consensus, to benign failure settings with flaky
channels. Roughly, a view synchronizer facilitates a commonly used design
pattern in which the protocol execution is divided into {\em views} (aka {\em
  rounds}). Each view has a designated leader that coordinates the process
interactions within that view.
The task of the synchronizer is to ensure that sufficiently many correct
processes are eventually able to spend sufficient time in the same view with a
correct leader, which would then be able to drive consensus to
completion. Supporting this functionality under partial synchrony is nontrivial,
as during asynchronous periods clocks can diverge and messages used to
synchronize views could get lost or delayed.
View synchronizers encapsulate the necessary logic to deal with these complexities, 
thereby enabling modular design of consensus protocols.

In contrast to failure or leader detectors, a synchronizer does not attempt to
identify the set of misbehaving processes (which, as we argued above, is
impossible with flaky channels), but instead delegates this task
to the top-level protocol. This can monitor the current leader using timeouts
and other protocol-specific logic and ask the synchronizer to switch to another
view with a different leader if it detects a lack of progress.
For example, if processes $1$ and $3$ in
Figure~\ref{fig:examples}c do not observe progress on behalf of 
process $2$ due to message loss, they would eventually initiate a view change
and try another leader.

We present a specification of a view synchronizer sufficient to implement
consensus in the presence of flaky channels, its implementation, and prove that
the implementation satisfies our specification.
Even though handling crashes is easier than Byzantine faults, flaky channels create
another challenge: processes outside the connected core 
(such as process $2$ in Figures~\ref{fig:examples}b-c)
falsely suspecting the current leader should not be able to force 
a view change, disrupting a working view.
Using our synchronizer we then design a consensus protocol that tolerates flaky
channels and prove its correctness.

Finally, we also demonstrate how our lower and upper bounds can be applied to
establish consensus solvability in various existing models of weak
connectivity~\cite{aguilera-podc03-journal,aguilera-podc04,dahlia-t-accessible,dahila-t-moving-source,antonio-intermittent-star}. In
particular, we show that some connectivity models strong enough to implement the
$\Omega$ leader detector are nevertheless too weak to implement consensus
(\S\ref{sec:cons-via-omega}).

\section{System Model}
\label{sec:model}

We consider a set $\PP$ of $n$ processes which can fail by crashing.
A process is {\em correct}\/ if it never crashes, and {\em faulty}\/ otherwise.
Processes communicate by exchanging messages through a set of unidirectional
channels $\CC$: for every pair of processes $p, q \in \PP$ there is a channel
$(p, q) \in \CC$ for sending messages from $p$ to $q$.

We consider two notions of channel correctness ({\em reliable} and {\em
  eventually reliable}) and two notions of channel faultiness ({\em
  disconnected} and {\em flaky}). Given correct processes $p$ and $q$:
\begin{itemize}
\item a channel $(p, q)$ is {\em reliable} if it delivers every message sent by
  $p$ to $q$;
\item a channel $(p, q)$ is {\em eventually reliable} if there exists a time $t$
  such that the channel delivers every message sent by $p$ to $q$ after $t$;
\item a channel $(p, q)$ is {\em disconnected} if it drops all messages sent by
  $p$ to $q$; and
\item a channel $(p, q)$ is {\em flaky} if it drops an arbitrary subset of
  messages sent by $p$ to $q$.
\end{itemize}
Flakiness is a very broad failure mode: it subsumes disconnections and allows
the channel to choose which messages to drop, without any guarantee of
fairness. Thus, flaky channels are strictly more permissive than fair lossy,
which are computationally equivalent to reliable, as we explained in
\S\ref{sec:intro}.

Our lower bounds assume the stronger notion of channel correctness (reliable)
and the more restrictive notion of faultiness (disconnected), whereas our upper
bounds assume weaker correctness (eventually reliable) and broader faultiness
(flaky). We combine these channel reliability assumptions with two types of
synchrony guarantees for correct channels, which yields four models used in our
results (Figure~\ref{fig:models}). We use the {\em asynchronous} model in our
results about registers, and the {\em partially synchronous}
model~\cite{dls,CT96} in our results about consensus. The latter assumes that
there exists a {\em global stabilization time} $\GST$ and a duration $\delta$
such that after $\GST$ message delays between correct processes connected by
correct channels are bounded by $\delta$. However, messages sent before $\GST$
can get arbitrarily delayed. Partial synchrony also assumes that processes have
clocks that can drift unboundedly from the real time before $\GST$, but do not
drift thereafter. Both $\GST$ and $\delta$ are unknown to our protocols.

To state our results we need an ability to restrict which processes and channels
can fail. We do this by generalizing the classical notion of a {\em fail-prone
  system}~\cite{bqs}, which we formulate in a generic fashion irrespective of a
specific channel failure type (flaky or disconnected). We specify the failure
type when stating our lower or upper bounds. A {\em failure pattern} is a pair
$(P, C) \in 2^\PP \times 2^\CC$ that defines which processes and channels
fail in a single execution. We assume that $C$ only contains channels
between pairs of correct processes, since those incident to faulty processes
fail automatically: $(p,q) \in C \implies \{p,q\} \cap P = \emptyset$.  For a
failure pattern $f=(P, C)$, an execution $\sigma$ of the system is
$f$-\emph{compliant} if exactly the processes in $P$ and the channels in $C$
fail in $\sigma$. A {\em fail-prone system} $\FS$ is a set of failure patterns.
For example, a standard fail-prone system where any minority of processes can
fail and channels between correct processes cannot fail corresponds to
$\KFS=\{(Q,\emptyset) \mid Q \subseteq \mathcal{P} \wedge |Q| \le \lfloor
\frac{n-1}{2} \rfloor\}$.

\section{Correctness Properties}
\label{sec:correctness}

We consider algorithms $\ALG$ that implement shared memory objects
$\mathcal{O}$, such as a register or consensus, in the models just introduced.
For an arbitrary domain of values $\mathit{Val}$, the \emph{register} interface
consists of operations $\mathit{write}(v)$, $v\in \mathit{Val}$ and
$\mathit{read}()$, which return $\mathit{ack}$ and a value in $\mathit{Val}$,
respectively. The interface of a consensus object consists of a single operation
$\mathit{propose}(v)$, $v\in \mathit{Val}$, which returns a value in
$\mathit{Val}$.

\subparagraph{Safety.}
The consensus object is required to satisfy the standard safety properties:
(\emph{Agreement}) all terminating $\mathit{propose}$ invocations must return
the same value; and (\emph{Validity}) $\mathit{propose}$ can only return a value
passed to some $\mathit{propose}$ invocation.

We now specify the safety properties for registers.
An operation $\mathit{op'}$ {\em follows}\/ an operation $\mathit{op}$, denoted
$\mathit{op} \rightarrow \mathit{op'}$, if $\mathit{op'}$ is invoked after
$\mathit{op}$ returns; $\mathit{op'}$ is {\em concurrent}\/ with $\mathit{op}$
if neither $\mathit{op} \rightarrow \mathit{op'}$ nor
$\mathit{op'} \rightarrow \mathit{op}$. The $\mathcal{O}$'s {\em sequential
  specification}\/ specifies the behavior of $\mathcal{O}$ in the executions in
which no operations are concurrent with each other.
The sequential specification of a register requires every read to return the
value written by the latest preceding write operation.  
An execution $\sigma$ of $\ALG$ is \emph{linearizable}~\cite{linear} if there
exists a set of responses $X$ and a sequence
$\pi = \mathit{op}_1,\mathit{op}_2,\dots$ of all complete 
operations in $\sigma$ and some subset of incomplete operations paired with responses 
in $X$ such that $\pi$ complies with $\mathcal{O}$'s sequential specification and satisfies
$\mathit{op}_i \rightarrow \mathit{op}_j \implies i < j$.  An execution $\sigma$
of a register satisfies \emph{safeness} if the subsequence of $\sigma$
consisting of all write operations and all read operations not concurrent with
any writes is linearizable. An implementation $\ALG$ of a register is
\emph{safe}~\cite{lamport1985on} if each one of its executions satisfies
safeness; $\ALG$ is \emph{atomic} if each one of its executions is
linearizable.

\subparagraph{Liveness.} Due to faulty channels, some correct processes may not
be able to terminate, e.g., if they are partitioned off from the rest of the
system. Therefore, to state liveness we parameterize the classical notions of
obstruction-freedom and wait-freedom by the failures allowed and the subsets of
correct processes where termination is required. We use the weaker
obstruction-freedom in our lower bounds and the stronger wait-freedom in our
upper bounds.  

For a failure pattern $f=(P,C)$, we say that $\ALG$ is
$(f,T)$-\emph{wait-free} if $T \subseteq \PP \setminus P$ and for every process
$p \in T$, operation $\mathit{op}$ and $f$-compliant fair execution $\sigma$ of
$\ALG$, if $\mathit{op}$ is invoked by $p$ in $\sigma$, then $\mathit{op}$
eventually returns. An operation $\mathit{op}$ 
\emph{eventually executes solo} in an execution $\sigma$ if either \emph{(i)} $op$ 
returns in $\sigma$, or \emph{(ii)} there exists a suffix $\sigma'$ of $\sigma$
such that for all operations $op'$, if $op'$ is concurrent with $op$ in $\sigma$,
then the process that invoked $op'$ is crashed in $\sigma'$.
We say that $\ALG$ is $(f,T)$-\emph{obstruction-free} if
$T \subseteq \PP \setminus P$ and
for every process $p \in T$, operation $\mathit{op}$ and $f$-compliant fair
execution $\sigma$ of $\ALG$, if $\mathit{op}$ is invoked by $p$ and 
eventually executes solo 
in $\sigma$, then $\mathit{op}$ eventually returns. 
Note that our notion of obstruction-freedom mirrors its well-known
shared memory counterparts, such as solo termination~\cite{solo-orig}
and the formalization of obstruction-freedom~\cite{obf-orig}
given in~\cite{obf-formal}.

We next lift these notions to a fail-prone system $\FS$ and a {\em termination
  mapping} $\tau: \FS \rightarrow 2^\PP$ -- a function mapping each failure
pattern to the set of correct processes whose operations are required to
terminate (thus,
$\forall f=(P, C) \in \FS.\ \tau(f) \subseteq \PP \setminus P$). Namely, we say
that $\ALG$ is $(\FS, \tau)$-\emph{obstruction-free} (respectively,
\emph{wait-free}) if for every $f \in \FS$, $\ALG$ is
$(f,\tau(f))$-obstruction-free (respectively, wait-free).  For example, the
standard guarantee of wait-freedom under a minority of processes failures
corresponds to $(\KFS, \tau_M)$-wait-freedom, where $\KFS$ is defined in
\S\ref{sec:model} and $\tau_M$ is such that
$\forall f=(P,\emptyset)\in \KFS.\, \tau_M(f) = \mathcal{P} \setminus P$.

\section{Inherent Connectivity Requirements for Registers and Consensus}
\label{section:lower_bound}

We now investigate which connectivity requirements are necessary to implement
first registers and then consensus in the models of Figure~\ref{fig:models}. We
start with a few simple results that serve as a stepping stone for our later lower
bounds. These results are also of independent interest because they generalize the
celebrated CAP theorem~\cite{brewer}: it is impossible to guarantee all of
Consistency, Availability and network Partition-tolerance.

\subsection{CAP Revisited}
\label{sec:cap-revisited}

The CAP formalization by Gilbert and Lynch~\cite{cap} considers an asynchronous
system without process failures. It models consistency as implementing an atomic
register, availability as providing wait-freedom at all processes, and
partition-tolerance as tolerating channels that can lose any messages (in our
terminology, flaky).

\begin{theorem}[CAP]
  \label{thm:cap}
  It is impossible to implement a wait-free atomic register in a message-passing system where
  processes do not fail, but all channels are flaky.
\end{theorem}

Despite its fame, the CAP theorem yields only a weak lower bound. First, it is
not tight: the theorem only says that, to implement a wait-free atomic register,
{\em some} channels must be correct, but obviously, a single correct channel
would not be sufficient in general. Second, CAP requires availability at {\em
  all} processes and thus is not applicable if we are trying to ensure
availability at least in {\em a part} of the system, as formalized by the
notions of liveness in \S\ref{sec:correctness}. Our first result lifts these
limitations and also strengthens CAP in other ways: it considers a weaker
object, obstruction-free safe register, and a stronger system model $\MRL$,
where channels can only fail by disconnection. Informally, we show that, in this
setting, all processes where obstruction-freedom holds must be transitively
connected via correct channels: no such process can be partitioned off from the
rest.

Formally, let $\GG=(\PP,\CC)$ be the directed graph constructed by taking
processes as vertices and channels as edges. For a failure pattern $f = (P, C)$,
let $\RG{\GG}{f}$ denote the subgraph of $\GG$ obtained by removing all
processes in $P$ along with their incident channels, as well as all channels in
$C$.
\begin{theorem}
\label{thm:1}
Let $f$ be a failure pattern and $T \subseteq \PP$. If some algorithm $\ALG$
implements an $(f,T)$-obstruction-free safe register over $\MRL$ (asynchronous /
reliable / disconnected), then $T$ is strongly connected in $\RG{\GG}{f}$.
\end{theorem}

Before proving the theorem, we note some of its consequences.
The theorem yields the following corollary in the special case when the failure
pattern $f$ disallows process failures and obstruction-freedom is required at
all processes.%
\begin{corollary}
  \label{thm:our-cap}
  Let $f$ be a failure pattern that disallows process failures:
  $\exists C \subseteq \CC.\, f = (\emptyset, C)$. If some algorithm $\ALG$
  implements an $(f,\PP)$-obstruction-free safe register over $\MRL$
  (asynchronous / reliable / disconnected), then the graph $\GG \setminus f$ is
  strongly connected.
\end{corollary}

This implies CAP (Theorem~\ref{thm:cap}): if an algorithm $\ALG$ implements an
atomic wait-free register, then it also implements an obstruction-free safe
register, and by the above corollary, all processes must be strongly
connected by correct channels.
Corollary~\ref{thm:our-cap} also yields a tight lower bound. To see this,
consider its lifting to an arbitrary fail-prone system $\FS$ and the termination
mapping $\tau_\PP = \lambda f.\, \PP$:
\begin{corollary}
  \label{thm:our-cap2}
  Let $\FS$ be a fail-prone system that disallows process failures:
  $\forall f \in \FS.\ \exists C \subseteq \CC.\, f = (\emptyset, C)$. If some algorithm
  $\ALG$ implements an $(\FS, \tau_\PP)$-obstruction-free safe register over
  $\MRL$ (asynchronous / reliable / disconnected), then for all $f \in \FS$ the
  graph $\GG \setminus f$ is strongly connected.
\end{corollary}

Then the following proposition implies a matching upper bound: it shows that a
wait-free atomic register can be implemented in the more adversarial model
$\MRU$ (asynchronous / eventually reliable / flaky).  The proposition follows
from a more general result we present later (Theorem~\ref{thm:upper-reg}).
\begin{proposition}
  Let $\FS$ be a fail-prone system that disallows process failures and such that
  for all $f \in \FS$, the graph $\GG \setminus f$ is strongly connected. Then
  there exists an algorithm $\ALG$ implementing an $(\FS, \tau_\PP)$-wait-free
  atomic register over $\MRU$ (asynchronous / eventually reliable / flaky).
\end{proposition}

\begin{proof}[Proof of Theorem~\ref{thm:1}]
Assume by contradiction that $\ALG$ is an $(f,T)$-obstruction-free implementation
of a safe register, but $T$ is not strongly connected in $\RG{\GG}{f}$.  Then
for some $u, v \in T$, either there is no path from $u$ to $v$ or from $v$ to
$u$ in $\RG{\GG}{f}$.  Without loss of generality we assume the former.  Let
$S_u$ be the strongly connected component
of $\RG{\GG}{f}$ containing $u$, and $S_v$ the strongly connected component of
$\RG{\GG}{f}$ containing $v$; then $S_u \cap S_v= \emptyset$.  Let $R_u$ be the
set of processes outside $S_u$ that can reach $u$ in $\RG{\GG}{f}$, and $R_v$
the set of processes outside $S_v$ that can reach $v$ (see
Figure~\ref{fig:diagram_lower}).

\begin{figure}[t]
  \centering
  \includegraphics[width=0.35\linewidth,trim=550 225 550 150,clip]{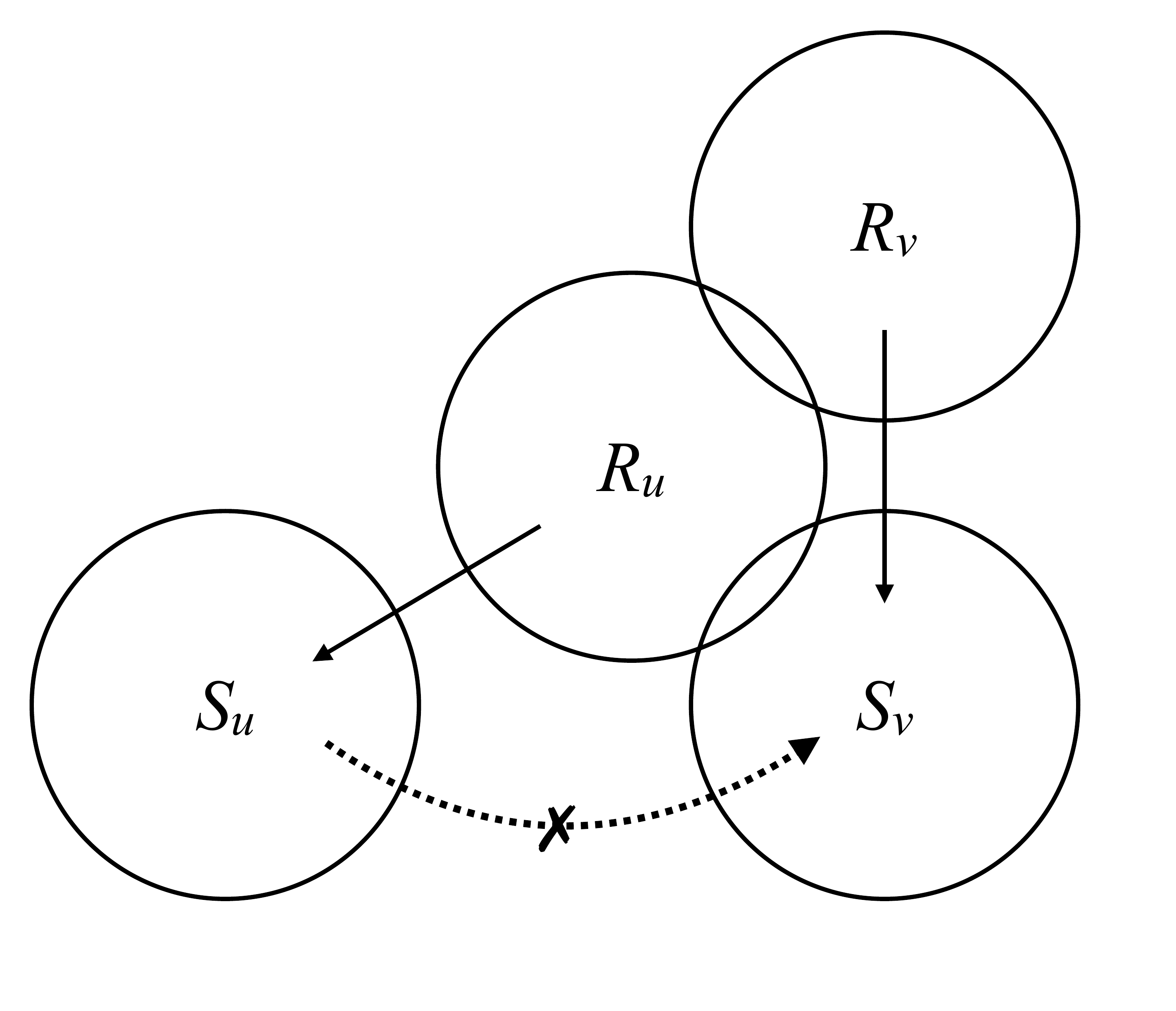}
  \caption{Illustration of the sets used in the proof of Theorem~\ref{thm:1}.}
  \label{fig:diagram_lower}
  \vspace{-3px}
\end{figure}

\setcounter{myclaim}{0}
\begin{myclaim}
\label{claim:8s3e}
For any $k \in \{u, v\}$, $R_k \cup S_k$ is unreachable from
$\PP \setminus (R_k \cup S_k)$ in $\RG{\GG}{f}$.
\end{myclaim}

\begin{myclaim}
\label{claim:1gc9}
For any $k \in \{u,v\}$, $R_k$ is unreachable from $S_k$ in $\RG{\GG}{f}$.
\end{myclaim}

The above claims easily follow from the definitions of $R_k$ and $S_k$.

\begin{myclaim}
$S_u \cap (R_v \cup S_v) = \emptyset$.
\end{myclaim}
\begin{claimproof}
  Assume by contradiction that $S_u \cap (R_v \cup S_v) \neq \emptyset$ and let
  $w \in S_u \cap (R_v \cup S_v)$.  Because $w \in S_u$, there exists a path
  from $u$ to $w$.  Because $w \in R_v \cup S_v$, there exists a path from $w$
  to $v$.  By concatenating these two paths we get a path from $u$ to $v$ in
  $\RG{\GG}{f}$, contradicting our assumption that there is no such path.
\end{claimproof}

Let $\alpha_1$ be a fair execution of $\ALG$ where processes and channels in $f$
fail at the beginning, the process $u$ invokes a \RD operation, and no other
operation is invoked in $\alpha_1$. Because $u \in T$ and $\ALG$ is
$(f,T)$-obstruction-free, the \RD 
operation must eventually terminate. Since there are no {\it write} invocations,
the \RD must return $0$ -- the initial value of the register. Let $\alpha_2$ be
the prefix of $\alpha_1$ ending with this response. By Claim~\ref{claim:8s3e},
$R_u \cup S_u$ is unreachable from $\PP \setminus (R_u \cup S_u)$. Thus, the
actions by processes in $R_u \cup S_u$ do not depend on those by processes in
$\PP \setminus (R_u \cup S_u)$. Then the projection of $\alpha_2$ to actions by
processes in $R_u \cup S_u$ is an execution of $\ALG$: $\alpha_3 = \alpha_2|_{R_u \cup
  S_u}$. By Claim~\ref{claim:1gc9}, $R_u$ is unreachable from $S_u$. Thus, the
actions by processes in $R_u$ do not depend on those by processes in $S_u$. Then
$\alpha = \alpha_3|_{R_u}\alpha_3|_{S_u}$ is also an execution of $\ALG$. Notice
that $\alpha_3|_{R_u}$ does not contain any operation invocation, but it
contains steps taken by processes upon startup.

Let $\beta_1$ be a fair execution of $\ALG$ where processes and channels in $f$
fail at the beginning. The execution $\beta_1$ starts with all the actions from
$\alpha_3|_{R_u}$ followed by a ${\it write}(1)$ invocation by process $v$, and
no other operation is invoked in $\beta_1$. Because $v \in T$ and $\ALG$ is
$(f,T)$-obstruction-free, the {\it write} operation must eventually terminate.
Let $\beta_2$ be the prefix of $\beta_1$ ending with this response. By
Claim~\ref{claim:8s3e}, $R_u \cup S_u$ is unreachable from 
$\PP \setminus (R_u \cup S_u)$. By Claim~\ref{claim:1gc9}, $R_u$ is unreachable
from $S_u$. Therefore, $R_u$ is unreachable from $\PP \setminus R_u$.  By
Claim~\ref{claim:8s3e}, $R_v \cup S_v$ is unreachable from
$\PP \setminus (R_v \cup S_v)$. Therefore, $R_u \cup R_v \cup S_v$ is
unreachable from $\PP \setminus (R_u \cup R_v \cup S_v)$.  Thus, the actions by
processes in $R_u \cup R_v \cup S_v$ do not depend on those by processes in
$\PP \setminus (R_u \cup R_v \cup S_v)$.  Then
$\beta = \beta_2|_{R_u \cup R_v \cup S_v}$ is an execution of
$\ALG$. Recall that $\beta_1$ starts with $\alpha_3|_{R_u}$, and hence, so does
$\beta$. Let $\delta$ be the suffix of $\beta$ such that
$\beta = \alpha_3|_{R_u}\delta$.

Consider the execution $\sigma = \alpha_3|_{R_u}\delta\alpha_3|_{S_u}$ where
the actions occur before processes and channels in $f$ fail. 
By Claim 3, $S_u \cap (R_v \cup S_v) = \emptyset$, and by the
definition of $R_u$, we have $S_u \cap R_u = \emptyset$.  Hence,
$S_u \cap (R_u \cup R_v \cup S_v) = \emptyset$.  Then, given that $\delta$
only contains actions by processes in $R_u \cup R_v \cup S_v$, we get:
$\sigma|_{R_u \cup R_v \cup S_v} =(\alpha_3|_{R_u}\delta\alpha_3|_{S_u})|_{R_u
  \cup R_v \cup S_v} = \alpha_3|_{R_u}\delta = \beta$ and
$\sigma|_{S_u} = (\alpha_3|_{R_u}\delta\alpha_3|_{S_u})|_{S_u} = \alpha_3|_{S_u}
= (\alpha_3|_{R_u}\alpha_3|_{S_u})|_{S_u}= \alpha|_{S_u}$.  Therefore, $\sigma$
is indistinguishable from $\beta$ to the processes in $R_u \cup R_v \cup S_v$
and from $\alpha$ to the processes in $S_u$. Finally,
$\sigma|_{\PP\setminus(R_u \cup S_u\cup R_v \cup S_v)} = \varepsilon$.

Thus, for every process, $\sigma$ is indistinguishable to this process from some
execution of $\ALG$. Furthermore, each message received by a process in $\sigma$
has previously been sent by another process. Therefore, $\sigma$ is an execution
of $\ALG$. However, in this execution
${\it write}(1)$ terminates before a \RD that fetches $0$ is invoked.  This contradicts the
assumption that $\ALG$ implements a safe register.  The contradiction shows that
$T$ must be strongly connected in $\RG{\GG}{f}$.
\end{proof}

\subsection{Connectivity Requirements under Bounded Process Failures}
\label{sec:cardinality}

Consider a straightforward lifting of Theorem~\ref{thm:1} to fail-prone systems:
\begin{corollary}
  \label{thm:1-lifted}
  Let $\FS$ be a fail-prone system and $\tau: \FS \rightarrow 2^\PP$ a
  termination mapping. If some algorithm $\ALG$ implements an
  $(\FS,\tau)$-obstruction-free safe register over $\MRL$ (asynchronous /
  reliable / disconnected), then for all $f \in \FS$, $\tau(f)$ is strongly
  connected in $\RG{\GG}{f}$.
\end{corollary}
This result does not make any assumptions about the fail-prone system, and its
CAP-like specialization given by Corollary~\ref{thm:our-cap2} only considers
fail-prone systems without process failures. However, algorithms are commonly
designed to tolerate a bounded number of process failures, and we next show that
stronger connectivity requirements are necessary in this case. To formalize this
class of algorithms, we use the following notion. A \emph{$k$-fail-prone system}
$\FS$ allows any set of $k$ processes or fewer to fail, but disallows failures
of more than $k$ processes:
\[
(\forall P \subseteq \PP.\ |P| \le k {\implies} \exists C \subseteq \CC.\ (P,C) \in \FS)
\wedge (\forall (P,C) \in \FS.\ |P| \le k).
\]
For example, the system $\KFS$ from \S\ref{sec:model} is
$\lfloor \frac{n-1}{2} \rfloor$-fail-prone. 

The following theorem establishes minimal connectivity constraints required to
implement a safe register under the model
$\MRL$ %
in the presence of $k$ process crashes. It assumes a particularly weak
termination guarantee which only requires obstruction-freedom to hold at some
\emph{non-empty} set of processes for each failure pattern.
The theorem states that, no matter how small this set
is, it must be part of a set of $>k$ correct processes
strongly connected by correct channels.

\begin{theorem}
\label{thm:2}
Let $\FS$ be a $k$-fail-prone system and $\tau: \FS \rightarrow 2^\PP$ a
termination mapping such that $\forall f=(P,C) \in \FS.\,
\tau(f)\neq\emptyset$. Assume that some algorithm $\ALG$ implements an
$(\FS,\tau)$-obstruction-free safe register over $\MRL$ (asynchronous / reliable
/ disconnected). Then for all $f\in\FS$, there exists a strongly connected
component of $\RG{\GG}{f}$ that contains $\tau(f)$ and has a cardinality greater
than $k$.
\end{theorem}
\begin{proof}
Let $f \in \FS$ and $S$ be the strongly connected component of $\RG{\GG}{f}$ containing $\tau(f)$.
Assume by contradiction that $|S|\leq k$.
Pick an arbitrary process $p \in \tau(f)$; then $p$ is correct according to $f$.
Let $\alpha_1$ be a fair execution of $\ALG$ where
processes and channels in $f$ fail at the beginning, the process $p$ invokes
a \RD operation, and no other operation is invoked in $\alpha_1$.
Because $p \in \tau(f)$ and $\ALG$ is $(\FS, \tau)$-obstruction-free, the \RD
operation must eventually terminate.
Since there are no {\it write} invocations, the \RD must return $0$ -- the initial value
of the register.
Let $\alpha_2$ be the prefix of $\alpha_1$ ending with this response.
Let $R$ be the set of processes outside $S$ that can reach $S$ in $\RG{\GG}{f}$.
Similarly to the proof of Theorem~\ref{thm:1}, the definitions of $R$ and $S$
imply the following claims.

\setcounter{myclaim}{0}
\begin{myclaim}
\label{claim:thm_2_1}
$R \cup S$ is unreachable from $\PP \setminus (R \cup S)$ in $\RG{\GG}{f}$.
\end{myclaim}
\begin{myclaim}
\label{claim:thm_2_2}
$R$ is unreachable from $S$ in $\RG{\GG}{f}$.
\end{myclaim}

Claim~\ref{claim:thm_2_1} implies that the actions by processes in $R \cup S$ 
do not depend on those by processes in $\PP \setminus (R \cup S)$.
Then $\alpha_3 = \alpha_2|_{R \cup S}$ is an execution of $\ALG$.
Claim~\ref{claim:thm_2_2} implies that the actions by processes in $R$ do not 
depend on those by processes in $S$.
Then $\alpha = \alpha_3|_{R}\alpha_3|_{S}$ is also an execution of $\ALG$.

Since $|S|\leq k$ and $\FS$ is a $k$-fail-prone
system, there exists $C' \subseteq \CC$ such that $f'=(S, C') \in \FS$.
Pick an arbitrary process $q \in \tau(f')$; then $q \not\in S$.
Let $\beta_1$ be a fair execution of $\ALG$
where processes and channels in $f'$ fail at the beginning.
The execution $\beta_1$ starts with all the actions from $\alpha_3|_{R}$
followed by a ${\it write}(1)$ invocation by process $q$, and no other operation
is invoked in $\beta_1$.
Because $q \in \tau(f')$ and $\ALG$ is $(\FS, \tau)$-obstruction-free, the {\it write} operation must
eventually terminate.
Let $\beta$ be the prefix of $\beta_1$ ending with this response and
let $\delta$ be the suffix of $\beta$ such that $\beta=\alpha_3|_{R}\delta$.

Consider the execution $\sigma = \alpha_3|_{R}\delta\alpha_3|_{S}$ where
the actions occur before processes and channels in $f$ fail. 
Given that $\delta$ does not contain any actions by processes in $S$, we get:
$\sigma|_{S} = \alpha_3|_{S} = (\alpha_3|_{R}\alpha_3|_{S})|_{S} = \alpha|_{S}$
and $\sigma|_{\PP \setminus S} = \alpha_3|_{R}\delta = \beta$.  Therefore,
$\sigma$ is indistinguishable from $\alpha$ to the processes in $S$ and from
$\beta$ to the processes in $\PP \setminus S$.

Thus, for every process, $\sigma$ is
indistinguishable to this process from some execution of $\ALG$. Furthermore,
each message received by a process in $\sigma$ has previously been sent by
another process. Therefore, $\sigma$ is an execution of $\ALG$. However, in this
execution ${\it write}(1)$ terminates before a \RD that fetches $0$ is invoked.  This
contradicts the assumption that $\ALG$ implements a safe register.  The
contradiction derives from assuming that $|S| \leq k$, so that $|S| > k$.
\end{proof}

Theorem~\ref{thm:2} is most interesting in the common practical case of
$n = 2k+1$, which is the minimal number of processes needed to tolerate $k$
crashes in asynchronous registers~\cite{lynch-dist-algos} and partially
synchronous consensus~\cite{dls}. In this case the theorem ensures that for each
failure pattern $f$, the graph $\RG{\GG}{f}$ has a strongly connected component
containing $\ge k+1$ processes. More generally, for arbitrary $\GG$ and $f$, we
call a strongly connected component of $\RG{\GG}{f}$ containing a majority of
processes in $\GG$ a {\em connected core} of the graph.  It is easy to see there
can exist at most one connected core for given $\GG$ and $f$. For example, in
Figure~\ref{fig:examples}a the connected core is $\{1, 2, 3\}$, whereas in
Figures~\ref{fig:examples}b-c it is $\{1, 3\}$.
As we now show, the lower bound of Theorem~\ref{thm:2} is tight for
$n = 2k+1$. In fact, assuming the existence of a connected core, we can
implement an atomic register that is wait-free at all members of the connected
core under the more adversarial model $\MRU$ (asynchronous / eventually reliable
/ flaky).

\begin{theorem}
  \label{thm:upper-reg}
  Let $\FS$ be a fail-prone system such that for all $f \in \FS$, the graph $\RG{\GG}{f}$ contains a
  connected core $\CQ_f$, and let $\tau: \FS \rightarrow 2^\PP$ be the termination
  mapping such that
  $\forall f \in \FS$.\ $\tau(f) = \CQ_f$. Then there exists an
  $(\FS,\tau)$-wait-free implementation of an atomic register over the model $\MRU$
  (asynchronous / eventually reliable / flaky).
\end{theorem}
We defer the proof of the theorem to
\tr{\ref{sec:app-abd}}{\appabd}. The proof constructs the desired implementation as a variant of
ABD~\cite{abd} that uses gossip-style data propagation to deal with indirect
connectivity. The most interesting aspect of this implementation is that,
despite the need for gossip, it uses only bounded space. We illustrate the
technique for bounding space when presenting our consensus implementation in
\S\ref{section:consensus}.

\subsection{Connectivity Requirements for Consensus}
\label{sec:consensus-lower}

The lower bounds in the previous section also apply to consensus under partial
synchrony, with analogs of Theorems~\ref{thm:1} and~\ref{thm:2} formulated as
follows:
\begin{theorem}
Let $f$ be a failure pattern and $T \subseteq \PP$. If some algorithm $\ALG$
is an $(f,T)$-obstruction-free implementation of consensus over $\MCL$ (partially
synchronous / reliable / disconnected), then $T$ is strongly connected in
$\RG{\GG}{f}$.
\label{thm:consensus-lower-1}
\end{theorem}
\begin{theorem}
  \label{thm:consensus-lower-2}
  Let $\FS$ be a $k$-fail-prone system and $\tau: \FS \rightarrow 2^\PP$ be a
  termination mapping such that
  $\forall f=(P,C) \in \FS.\ \tau(f)\neq\emptyset$. Assume that some algorithm
  $\ALG$ is an $(\FS,\tau)$-obstruction-free implementation of consensus over
  $\MCL$ (partially synchronous / reliable / disconnected).  Then for all
  $f\in\FS$, there exists a strongly connected component of $\RG{\GG}{f}$ that
  contains $\tau(f)$ and has a cardinality greater than $k$.
\label{thm:k-consensus-lower}
\end{theorem}

To see why these theorems hold
observe first that Theorems~\ref{thm:1} and~\ref{thm:2} also hold 
if the model $\MRL$ is replaced with $\MCL$: 
since the executions $\sigma$ constructed in the proofs of the theorems
are finite, they are also valid executions under $\MCL$
where all actions occur before $\GST$. Then the required follows from the fact that
registers can be implemented from consensus. We formally prove this for our setting in
\tr{\ref{sec:app-reg-from-consensus}}{\appregfromconsensus}.

Finally, similarly to Theorem~\ref{thm:upper-reg}, for
the case of $n = 2k+1$ we can prove an upper bound matching
Theorem~\ref{thm:k-consensus-lower} in the model $\MCU$ (partially synchronous /
eventually reliable / flaky). This result is much more difficult than the upper
bound for registers, and we prove it in the next section.
\begin{theorem}
  \label{thm:upper-consensus}
  Let $\FS$ be a fail-prone system such that for all $f \in \FS$, the graph $\RG{\GG}{f}$ contains a
  connected core $\CQ_f$, and let $\tau: \FS \rightarrow 2^\PP$ be the
  termination mapping such that
  $\forall f \in \FS.\ \tau(f) = \CQ_f$. Then there exists an
  $(\FS,\tau)$-wait-free implementation of consensus over the model $\MCU$
  (partially synchronous / eventually reliable / flaky).
\end{theorem}

\section{Consensus in the Presence of Flaky Channels}
\label{section:consensus}

We now present a consensus protocol in the model $\MCU$ that validates
Theorem~\ref{thm:upper-consensus}. Consider a fail-prone system $\FS$ satisfying
the conditions of the theorem: for each $f \in \FS$, the graph $\RG{\GG}{f}$
contains a connected core. For the remainder of the section we fix a failure
pattern $f \in \FS$ and let $\CQ$ be the corresponding connected core.
Since, as we explained in \S\ref{sec:intro}, classical failure and leader
detectors are not useful in the presence of flaky channels, we take a different
approach. Our protocol is implemented on top of a \emph{view
  synchronizer}~\cite{multishot-disc22,bftlive-dc,cogsworth,oded-linear}, which
enables the processes to divide their execution into a series of views, each
with a designated leader.
At a high level, the synchronizer's goal is to bring sufficiently many correct
processes with enough connectivity (e.g., those from $\CQ$) into a view led by a
well-connected leader and keep them in that view for sufficiently long to reach
an agreement.
Supporting this in the presence of flaky channels is nontrivial as the processes
outside the connected core (such as process $2$ in Figures~\ref{fig:examples}b-c)
may fail to observe progress on behalf of the leader of a functional view and
request a premature view change. We first present the specification
(\S\ref{sec:spec}) and the implementation (\S\ref{sec:syncimpl}) of a
synchronizer that addresses this challenge. We then use it to construct a
consensus protocol (\S\ref{section:consensus-protocol}) satisfying the
requirements of Theorem~\ref{thm:upper-consensus}.

\subsection{Synchronizer Specification}
\label{sec:spec}

We consider a synchronizer interface defined in
\cite{cogsworth,multishot-disc22}. Let $\mathsf{View} = \{1,2,\dots\}$ be the
set of {\it views}, ranged over by $v$; we use $0$ to denote an invalid initial
view. The synchronizer produces notifications $\newview(v)$ at a process,
telling it to enter a view $v$. To trigger these, the synchronizer allows a
process to call a function $\adv()$, which signals that the process wishes to
{\em advance} to a higher view. We assume that a process does not call $\adv()$
twice without an intervening $\newview$ notification.

\begin{figure}[t]
  \small
  
    \begin{itemize}
        \setlength\itemsep{4pt}

        \item {\bf Monotonicity.} A process may only enter increasing views:\\[2pt]
        $\forall i, v, v'.\ {E_i(v)\fdef}\ \land\ {E_i(v')\fdef} \implies (v < v' \iff E_i(v) < E_i(v'))$

        \item {\bf Validity.} A process only enters $v+1$ if some process from $\CQ$ has attempted
        to advance from $v$:\\[2pt]
        $\forall i, v.\ {E_i(v + 1)\fdef} \implies {\AF{\CQ}(v)\fdef}\ \land\ \AF{\CQ}(v) < E_i(v + 1)$

        \item {\bf Bounded Entry.} For some $\V$ and $d$, if a process from $\CQ$ enters 
        $v \geq \V$ and no process from $\CQ$ attempts to advance to a higher view within $d$,
        then every process from $\CQ$ will enter $v$ within $d$:\\[2pt]
        $\exists \V, d.\ \forall v \geq \V.\ {\EF{\CQ}(v)\fdef}\ \land\
        \neg(\AF{\CQ}(v) < \EF{\CQ}(v) + d) \implies$\\[2pt]
        $(\forall p_i \in \CQ.\ {E_i(v)\fdef})\ \land\ (\EL{\CQ}(v) \leq \EF{\CQ}(v) + d)$

        \item {\bf Startup.} If $>\frac{n}{2}$ processes from $\CQ$ invoke $\adv$, then some process from $\CQ$
        will enter view $1$:\\[2pt]
        $(\exists P \subseteq \CQ.\ |P| > \frac{n}{2}\ \land\
        (\forall p_i \in P.\ {A_i(0)\fdef})) \implies {\EF{\CQ}(1)\fdef}$ 

        \item {\bf Progress.} If a process from $\CQ$ enters $v$ and, for some set $P \subseteq \CQ$
        of $>\frac{n}{2}$ processes, any process in $P$ that enters $v$ eventually invokes $\adv$, then some process
        from $\CQ$ will enter $v+1$:\\[2pt]
        $\forall v.\ {\EF{\CQ}(v)\fdef}\ \land\ (\exists P \subseteq \CQ.\
        |P| > \frac{n}{2}\ \land\ (\forall p_i \in P.\ {E_i(v)\fdef} \implies
        {A_i(v)\fdef})) \implies {\EF{\CQ}(v + 1)\fdef}$ 
      \end{itemize}

{\bf Notation:}

\begin{itemize}
\setlength\itemsep{4pt}
\item    
$E_i(v)$: the time when process $p_i$ enters a view $v$
\item
$\EF{\CQ}(v)$, $\EL{\CQ}(v)$: the earliest and the latest time when a process
from $\CQ$ enters a view $v$ 
\item
$A_i(v)$, $\AF{\CQ}(v)$, $\AL{\CQ}(v)$:
similarly for times of attempts to advance from a view $v$
\item
$g(x)\fdef$, $g(x)\fndef$: $g(x)$ is defined/undefined
\end{itemize}
      
    \caption{Synchronizer properties satisfied in executions with connected core
      $\CQ$.}
    \label{fig:synchronizer_specification}
\end{figure}

In Figure~\ref{fig:synchronizer_specification} we give a specification of a view
synchronizer for the system model $\MCU$, which is an adaptation of the one for
the Byzantine setting~\cite{multishot-disc22}.
The Monotonicity property ensures that, at any given process, its view can only
increase.  The Validity property ensures that a process may only enter a view
$v + 1$ if some process in the connected core $\CQ$ has called $\adv$ in $v$.
This prevents processes with bad connectivity from disrupting $\CQ$ by forcing
view changes (e.g., process $2$ in Figure~\ref{fig:examples}c, where
$\CQ = \{1, 3\}$). The Bounded Entry property ensures that, if some process from
$\CQ$ enters a view $v$, then all processes from $\CQ$ will do so at most $d$
units of time of each other (for some constant $d$). This only holds if within
$d$ no process from $\CQ$ attempts to advance to a higher view, as this may make
some processes from $\CQ$ skip $v$ and enter a higher view directly. Bounded
Entry only holds starting from some view $\V$, since a synchronizer may not be
able to guarantee it for views entered before $\GST$. The Startup property
ensures that if more than $\frac{n}{2}$ processes from $\CQ$ attempt to advance
from view $0$, then some process from $\CQ$ will enter view $1$. The Progress
property determines the conditions under which some process from $\CQ$ will
enter a view $v + 1$: this will happen if a process from $\CQ$ enters the view $v$,
and for some set $P$ of more than $\frac{n}{2}$ processes from $\CQ$, any
process in $P$ entering $v$ eventually invokes $\adv$ (e.g., $1$ and
$3$ in Figure~\ref{fig:examples}c).

As we show below, the above properties work in tandem to ensure the liveness of
consensus in the presence of flaky channels. Informally, Progress allows
processes to iterate over views in search for one with a well-connected leader;
Bounded Entry enables all processes in the connected core to promptly enter this
view; and Validity ensures that these processes can stay in it despite any
disruption from processes with flaky connectivity.

\subsection{Synchronizer Implementation}
\label{sec:syncimpl}

In Figure~\ref{fig:synchronizer_implementation} we present an algorithm that
implements the specification in Figure~\ref{fig:synchronizer_specification} in
the model $\MCU$. This implementation requires only bounded space, despite the
fact that correct channels in $\MCU$ are only {\em eventually} reliable, and
thus can lose messages before $\GST$. A process stores its current view in a
variable $\currview$. A process also maintains an array $\views$ tracking, for
every other process, the highest view to which it wishes to advance. When the
process invokes $\adv$ (line \ref{protocol:sync:adv}), the synchronizer does not
immediately switch to the next view. Instead, it updates its entry in the
$\views$ array and propagates the whole array in a $\WISH$ message, advertising
its wish to advance (line \ref{protocol:sync:adv:wish}).  Upon the receipt of a
$\WISH$ message (line \ref{protocol:sync:wish}), a process incorporates the
information received into its $\views$ array, keeping entries with the highest
view (line \ref{protocol:sync:wish:views}).  This mechanism ensures that
information is propagated between processes that do not have direct connectivity
via a correct channel.

\begin{figure}[t]
  \begin{algorithm}[H]
    \small
    \SetAlgoNoLine
    \DontPrintSemicolon      
    \setcounter{AlgoLine}{0}

    \Function{$\adv()$}{\label{protocol:sync:adv}
      \assign{\views[i]}{\currview+1}\;\label{protocol:sync:adv:views}
      \send $\WISH(\views)$ \ToAll\;\label{protocol:sync:adv:wish}
    }

    \smallskip
    \smallskip

    \SubAlgo{${\bf periodically}$\ \ \ \ \ \ \ $\triangleright$ every $\rho$ time units}{
      \label{protocol:sync:periodically}
      \send $\WISH(\views)$ \ToAll\;\label{protocol:sync:periodically:wish}
    }

    \smallskip
    \smallskip

    \SubAlgo{$\onreceive\ \WISH(V)$}{\label{protocol:sync:wish}
      \lFor{$p_j \in \mathcal{P}$}{\label{protocol:sync:wish:viewsfor}
        \assign{\views[j]}{\max(\views[j], V[j])}\label{protocol:sync:wish:views}
      }
      \assign{v'}{\max \{v \mid \exists p_j.\ \views[j] = v\, \land\,
        {|\{p_k \mid \views[k] \geq v\}| > \frac{n}{2}}\}}\;\label{protocol:sync:wish:v}
      \If{$v' > \currview$}{\label{protocol:sync:wish:guard}
        \assign{\currview}{v'}\;
        $\trigger\ \newview(v')$\;\label{protocol:sync:wish:newview}
        \send $\WISH(\views)$ \ToAll\;\label{protocol:sync:wish:wish}
      }
    }
  \end{algorithm}
  \caption{Synchronizer at a process $p_i$.}
  \label{fig:synchronizer_implementation}
\end{figure}
Since the membership of $\CQ$ is unknown to the processes, to satisfy Validity
the synchronizer cannot initiate a view change based on an $\adv()$ call by a
single process: the synchronizer cannot tell whether this process is from $\CQ$
or not. Instead, we require wishes to advance from a majority of processes, so
that at least one of them must be from $\CQ$. In more detail, upon receiving a
$\WISH$ message, we compute $v'$ as the $(\lfloor\frac{n}{2}\rfloor + 1)$-st
highest view in $\views$ (line \ref{protocol:sync:wish:v}). Thus, at least one
process from $\CQ$ wishes to advance to a view $\geq v'$. In this case the current
process enters $v'$ if this view is greater than its $\currview$ (line
\ref{protocol:sync:wish:newview}). Note that a process may be forced to switch
views even if it did not invoke $\adv$; this helps lagging processes to catch
up. To satisfy Bounded Entry, the process disseminates the information that made
it enter the new view (line \ref{protocol:sync:wish:wish}) to ensure that
other processes also do so promptly. Finally, to deal with message loss before
$\GST$, a process periodically resends its $\views$ array (line
\ref{protocol:sync:periodically}).

We can also prove that the synchronizer satisfies Progress: this property
requires $>\frac{n}{2}$ $\adv$ calls by processes in $\CQ$, which are
well-connected enough for the corresponding $\WISH$es to eventually propagate within
$\CQ$ and enable the guard at line~\ref{protocol:sync:wish:guard} (e.g., processes
$1$ and $3$ in Figure~\ref{fig:examples}c). Note that Progress wouldn't hold if
it required $>\frac{n}{2}$ $\adv$ calls by any processes, not necessarily in
$\CQ$ (e.g., processes $1$ and $2$): in this case we wouldn't be able to
guarantee that all the corresponding $\WISH$es eventually propagate.
We defer the proof of correctness of the synchronizer to
\tr{\ref{sec:app-sync}}{\appsync}.

\begin{theorem}\label{theorem:nxp}
  Let $\FS$ be a fail-prone system such that for each $f \in \FS$, $\RG{\GG}{f}$ contains a
  connected core. Then for any $f \in \FS$ with an associated connected
  core $\CQ$, every $f$-compliant fair execution of the algorithm in
  Figure~\ref{fig:synchronizer_implementation} over the model $\MCU$ satisfies the
  properties in Figure~\ref{fig:synchronizer_specification}.
\end{theorem}

\subsection{Consensus Protocol}
\label{section:consensus-protocol}

\begin{figure}[t]
	\begin{tabular}{@{\quad\ \ }l@{}|@{\quad \ \ }l@{}}
	\hspace{-21px}
	\begin{minipage}[t]{6.9cm}
		\vspace{-14px}
		\begin{algorithm}[H]
			\DontPrintSemicolon
			\SetAlgoNoLine
			\setcounter{AlgoLine}{0}
			\SetInd{0.5em}{0.5em}

			\SubAlgo{${\bf on\ startup}\xspace$}{
				$\adv$()\;\label{protocol:start:adv}
			}

			\smallskip
			\smallskip

				\Function{$\propose(x)$}{\label{protocol:consensus:propose}
					$\assign{\pval}{x}$\; \label{protocol:consensus:propose:assign}
			    \textbf{wait until}\xspace $\phase = \DECIDED$\;\label{protocol:consensus:propose:condition}
					\textbf{return}\xspace $\val$\;
			  }

	      \smallskip
	      \smallskip

			\SubAlgo{${\bf on }\xspace\ \newview(v)$}{
				\label{protocol:consensus:newview}
				$\assign{\view}{v}$\;
		    $\starttimer(\timerdecision, \durdecision)$\;\label{protocol:consensus:newview:starttimer}
		    $\assign{\MSGOB[i]}{(\view,\cview,\val)}$\;\label{protocol:consensus:newview:msg1b}
				$\assign{\phase}{\ENTERED}$\;\label{protocol:consensus:newview:phase}
			}

			\smallskip
			\smallskip

			\SubAlgo{{\bf when the timer $\timerdecision$ expires}}{\label{protocol:consensus:timerexpires}
				$\assign{\durdecision}{\durdecision + \gamma}$\;\label{protocol:consensus:timerexpires:durdecision}
				$\adv()$\;\label{protocol:consensus:timerexpires:adv}
			}

			\smallskip
			\smallskip

	    \SubAlgo{${\bf periodically}$\ \ \ \ \ \ \ $\triangleright$ every $\rho$ time units}{
	        \label{protocol:consensus:periodically}
	        \send $\STATE(\MSGOB,\MSGTA,\MSGTB)$ \ToAll\;\label{protocol:consensus:periodically:state}
	    }

			\smallskip
			\smallskip

			\SubAlgo{\onreceive $\STATE(\VOB, \VTA, \VTB)$}{
				\label{protocol:consensus:state}
				\For{$p_j \in \mathcal{P}$}{
					\If{$\VOB[j].\itview > \MSGOB[j].\itview$}{
						\assign{\MSGOB[j]}{\VOB[j]}\;\label{protocol:consensus:state:msgob}
					}
					\If{$\VTA[j].\itview > \MSGTA[j].\itview$}{
						\assign{\MSGTA[j]}{\VTA[j]}\;\label{protocol:consensus:state:msgta}
					}
					\If{$\VTB[j].\itview > \MSGTB[j].\itview$}{
						\assign{\MSGTB[j]}{\VTB[j]}\;\label{protocol:consensus:state:msgtb}
					}
        }
      }
		\end{algorithm}
				\end{minipage}
		\hspace{-5px}
&
		\hspace{-12px}
		\begin{minipage}[t]{7.5cm}
		\vspace{-14px}
		\begin{algorithm}[H]
      \DontPrintSemicolon
      \SetAlgoNoLine
      \SetInd{0.5em}{0.5em}

			\SubAlgo{\when $\phase=\ENTERED \land \leader(\view) =
                          p_i \land {}$\\
			 \nonl $|\{p_j\ |\ p_j \in \PP \land \MSGOB[j].\itview = \view\}| > \frac{n}{2}$}{
				\label{protocol:consensus:proposed}
      	\assign{Q}{\{p_j\ |\ p_j \in \PP \land \MSGOB[j].\itview = \view\}}\;
        \uIf{$\forall p_j.\ p_j \in Q {\implies} \MSGOB[j].\itval = \bot$}{\label{protocol:consensus:all_bot}
      		\lIf{$\pval = \bot$}{\textbf{return}}\label{protocol:consensus:return}
      		$\assign{\MSGTA[i]}{(\view,\pval)}$\;\label{protocol:consensus:msgta1}
				}\Else{
					{\bf let $p_j \in Q$ be such that}
                                        \linebreak $\phantom \ \ \ \MSGOB[j].\itval \neq \bot \land \linebreak
					\phantom \ \ \ \forall p_k \,{\in}\, Q.\, \MSGOB[k].\itcview \leq \MSGOB[j].\itcview$\;\label{protocol:consensus:pick}
					$\assign{\MSGTA[i]}{(\view,\MSGOB[j].\itval)}$\;\label{protocol:consensus:msgta2}
				}
				$\assign{\phase}{\PROPOSED}$\;\label{protocol:consensus:phase}
      }

	    \smallskip
			\smallskip
      \smallskip

      \SubAlgo{\when $\phase \in \{\ENTERED,\PROPOSED\} \land{}$\\
	    	\nonl $p_l=\leader(\view) \land \MSGTA[l].\itview =
            \view$}{
        \label{protocol:consensus:accepted}
				$\assign{(\cview,\val)}{\MSGTA[l]}$\;\label{protocol:consensus:accepted:cview_val}
				$\assign{\MSGTB[i]}{(\cview,\val)}$\;\label{protocol:consensus:accepted:msgtb}
				$\assign{\phase}{\ACCEPTED}$\;\label{protocol:consensus:accepted:phase}
      }

      \smallskip
      \smallskip
      \smallskip

	    \SubAlgo{\when $\exists v,x.\, v \geq \view \land{}$\\
	    	\nonl $|\{p_j\ |\ p_j \in \PP \land \MSGTB[j] = (v,x)\}| > \frac{n}{2}$}{
        \label{protocol:consensus:decided}
	    	\assign{\val}{x}\;\label{protocol:consensus:decided:val}
                $\stoptimer(\timerdecision)$\;\label{protocol:consensus:decided:stoptimer}
	    	\assign{\phase}{\DECIDED}\;\label{protocol:consensus:decided:phase}
              }
		\end{algorithm}
				\end{minipage}
	\end{tabular}
	\caption{Consensus protocol at a process $p_i$.}
	\label{fig:consensus_protocol}
\end{figure}

In Figure~\ref{fig:consensus_protocol} we present a consensus protocol in the
model $\MCU$ (partially synchronous / eventually reliable / flaky) that
validates Theorem~\ref{thm:upper-consensus}. The protocol is a variation of
single-decree Paxos~\cite{paxos} where liveness is ensured with the help of a
view synchronizer.
Thus, the protocol works in a succession of views produced by the synchronizer.
Each view $v$ has a fixed leader $\leader(v) = p_{((v-1) \mod n)+1}$
responsible for proposing a value to the other processes, which vote on the
proposal. Processes monitor the leader's behavior and ask the synchronizer to
advance to another view if they suspect that the leader is faulty or has a bad
connectivity.

\subparagraph{State and communication.}
A process stores its current view in a variable $\view$. A variable $\phase$ tracks
the progress of the process through different phases of the protocol. The
initial proposal is stored in $\pval$
(line~\ref{protocol:consensus:propose:assign}). The process also maintains the
last proposal it accepted from a leader in $\val$, and the view in which this
happened in $\cview$.

Processes exchange messages analogous to the $\mathsf{1B}$, $\mathsf{2A}$ and
$\mathsf{2B}$ messages from Paxos, each tagged with the view where the
message was issued. There is no analog of $\mathsf{1A}$ messages, because leader
election is controlled by the synchronizer. Since correct processes may not be
directly connected by correct channels, each process has to forward the
information it receives from others. Since correct channels in $\MCU$
are only {\em eventually} reliable, this furthermore has to be repeated
periodically. If implemented naively, this would require unbounded space to
store all the messages that need to be forwarded. Instead, we observe that it is
sufficient to store, for each message type and sender, only the message of this
type received from this sender with the highest view. These are stored in the
arrays $\MSGOB$, $\MSGTA$ and $\MSGTB$.

For simplicity, the pseudocode in Figure~\ref{fig:consensus_protocol} separates
computation from communication. Most of the handlers do not send messages, but
instead just modify the arrays $\MSGOB$, $\MSGTA$ and $\MSGTB$. Then instead of
sending individual $\mathsf{1B}$, $\mathsf{2A}$ or $\mathsf{2B}$ messages like
in Paxos, a process periodically sends whole arrays $\MSGOB$, $\MSGTA$ and
$\MSGTB$ in one big $\STATE$ message (line
\ref{protocol:consensus:periodically:state}). Upon receipt of a $\STATE$ message
(line \ref{protocol:consensus:state}), a process incorporates the information
received into its arrays $\MSGOB$, $\MSGTA$ and $\MSGTB$, keeping entries with
the highest view. This mechanism ensures that information is propagated between
processes that do not have direct connectivity while using only bounded space.

\subparagraph{Normal protocol operation.}  When the synchronizer tells a process to
enter a new view $v$ (line \ref{protocol:consensus:newview}), the process sets
$\view = v$ and writes the information about the last value it accepted into
its entry in the $\MSGOB$ array. This information will be propagated to the
leader of the view as described above.
A leader waits until its $\MSGOB$ array contains a majority of entries
corresponding to its view (line~\ref{protocol:consensus:proposed}). Based on
these, the leader computes its proposal and stores it into its entry of
the $\MSGTA$ array. The computation is done similarly to Paxos. If some process
has previously accepted a value, the leader picks the value accepted in the
maximal view (line~\ref{protocol:consensus:pick}). Otherwise, the leader is free
to propose its own value. If $\propose()$ has already been invoked at the
leader, it selects $\pval$ (line~\ref{protocol:consensus:msgta1}). If this has
not happened yet, the leader skips its turn
(line~\ref{protocol:consensus:return}).

Each process waits until its $\MSGTA$ array contains a proposal by the leader of
its view (line~\ref{protocol:consensus:accepted}). The
process then accepts the proposal by updating its $\val$ and $\cview$. It also
notifies all other processes about this by storing the information about the
accepted value into its entry of the $\MSGTB$ array.
Finally, once a process has a majority of matching entries in its $\MSGTB$ array
(line \ref{protocol:consensus:decided}), it knows that the decision has been
reached, and it sets $\phase =\DECIDED$. If there is an ongoing $\propose()$
invocation, the condition at line \ref{protocol:consensus:propose:condition} is
then satisfied and the process returns the decision to the client.

\subparagraph{Triggering view changes.}
We now describe when a process invokes $\adv()$, which is key to ensuring
liveness. This occurs either on startup (line~\ref{protocol:start:adv}) or when
the process suspects that the current leader is faulty or has bad
connectivity. To this end, when a process enters a view, it sets a
$\timerdecision$ for a duration $\durdecision$ (line
\ref{protocol:consensus:newview:starttimer}) and stops the timer when a decision is
reached (line \ref{protocol:consensus:decided:stoptimer}). If the timer expires
before this, the process invokes $\adv$ (line
\ref{protocol:consensus:timerexpires:adv}). A process may wrongly suspect a good
leader if the $\durdecision$ is initially set too low with respect to the
message delay $\delta$, unknown to the process. To deal with this, a process
increases $\durdecision$ whenever the timer expires (line
\ref{protocol:consensus:timerexpires:durdecision}).

\subparagraph{Correctness.}  It remains to prove that the algorithm in
Figure~\ref{fig:consensus_protocol} validates Theorem~\ref{thm:upper-consensus}.
The proof of the safety properties of consensus is virtually
identical to that of Paxos~\cite{paxos}. Hence, we focus on proving
liveness. Here we present a proof sketch that highlights the use of the
synchronizer specification and defer the proofs of auxiliary lemmas to
\tr{\ref{sec:app-consensus}}{\appconsensus}.

Fix a failure pattern $f$ and let $\CQ$ be the corresponding connected core
guaranteed to exist by the assumptions of Theorem~\ref{thm:upper-consensus}.  We
prove liveness by showing that the protocol establishes properties reminiscent
of those of failure detectors~\cite{CT96}.  First, similarly to their {\em
  completeness} property, we prove that every correct process eventually
attempts to advance from a view where no progress is possible (e.g., because the
leader is faulty or has insufficient connectivity). We say that a process $p_i$
{\em decides} in a view $v$ if it executes
line~\ref{protocol:consensus:decided:phase} while having $\view = v$.
\begin{lemma}
	\label{lemma:not_ok_adv}
	If a correct process $p_i$ enters a view $v$, never decides in $v$ and
  never enters a view higher than $v$, then $p_i$ eventually invokes
  $\adv$ in $v$.
\end{lemma}
Informally, the lemma holds because each process monitors the progress of a view
using $\timerdecision$.

Our next lemma is similar to the {\em eventual accuracy} property of failure
detectors. It shows that if the timeout values are high enough, then eventually
any process in $\CQ$ that enters a view where progress is possible (the leader
is correct and has sufficient connectivity) will not attempt to advance from it.
Formally, let $\diameter(\CQ)$ be the longest distance in the graph
$\RG{\GG}{f}$ between two vertices in $\CQ$. Also, let $\V$ and $d$ be the view
and the time duration for which Bounded Entry holds
(Figure~\ref{fig:synchronizer_specification}), and let $\durdecision_i(v)$ be
the value of $\durdecision$ at the process $p_i$ while in view $v$.
\begin{lemma}
	\label{lemma:ok_stay_all}
	Let $v \geq \V$ be a view such that $\leader(v) \in \CQ$,
        $\EF{\CQ}(v) \geq \GST$ and $\leader(v)$ invokes $\propose$ no later
        than $\EF{\CQ}(v)$.  If at each process $p_i \in \CQ$ that enters $v$ we
        have $\durdecision_i(v) > d + 3 (\delta + \rho) \diameter(\CQ)$, then no
        process in $\CQ$ invokes $\adv$ in $v$.
\end{lemma}
Intuitively, the lemma holds because Bounded Entry ensures that all processes
from $\CQ$ will enter $v$ promptly; then since the timeouts are high enough,
processes will have sufficient time to exchange the messages needed to reach a
decision and stop the timers.

\begin{proof}[Proof sketch for Theorem~\ref{thm:upper-consensus}]
  By contradiction, assume there exists $f \in \FS$, $p_j \in \tau(f) = \CQ$ and
  an $f$-compliant fair execution of the algorithm in
  Figure~\ref{fig:consensus_protocol} such that $p_j$ invokes $\propose$ at a
  time $t$, but the operation never returns. Using Progress and
  Lemma~\ref{lemma:not_ok_adv}, we first prove that in this case the protocol
  keeps moving through views forever:
  \setcounter{myclaim}{0}
\begin{myclaim}
\label{myclaim:all_views_entered}
Every view is entered by some process in $\CQ$.
\end{myclaim}

We next prove the following:
\begin{myclaim}
\label{myclaim:exp_inf_times}
Every process in $\CQ$ executes the timer expiration handler at
line \ref{protocol:consensus:timerexpires} infinitely often.
\end{myclaim}

Since a process increases $\durdecision$ every time $\timerdecision$
expires, by Claim~\ref{myclaim:exp_inf_times} all processes will eventually 
have $\durdecision > d + 3 (\delta + \rho) \diameter(\CQ)$. Since leaders rotate
round-robin, by Claim 1 there will be infinitely many views led by $p_j$.
Hence, there exists a view $v_1 \geq \V$ led by $p_j$ such that
$\EF{\CQ}(v_1) \geq \max\{\GST,t\}$ and for any process $p_i \in \CQ$ that
enters $v_1$ we have
$\durdecision_i(v_1) > d + 3 (\delta + \rho) \diameter(\CQ)$.  Then by
Lemma~\ref{lemma:ok_stay_all}, no process in $\CQ$ calls $\adv$ in $v_1$. On the
other hand, by Claim 1 some process in $\CQ$ enters $v_1 + 1$.  Then by
Validity, some process in $\CQ$ calls $\adv$ in $v_1$, which is a contradiction.
\end{proof}

\section{Possibility and Impossibility of Classical Consensus via $\mathbf{\Omega}$}
\label{sec:cons-via-omega}

We now present some interesting consequences of our results.
It is well-known that the failure detector $\Omega$~\cite{CT96,CT96-weakest} 
is sufficient for implementing wait-free consensus resilient to $\lfloor \frac{n-1}{2} \rfloor$ failures 
in non-partitionable systems under crash~\cite{CT96,paxos},
crash/recovery~\cite{crash-rec-consensus-AguileraCT00},  
or general omission~\cite{paxos-with-omissions,paxos} failures.
We now investigate if this is still the case
under flaky channels and limited connectivity. Below we demonstrate that the answer
depends on the amount of connectivity 
in the underlying network, as established by our lower and upper bounds.

\subparagraph{Lower bound.}  We first show that, on the negative side, there are
some weakly synchronous environments where $\Omega$ is implementable, but
obstruction-free consensus is impossible even if at most one process can
crash. One such environment is the system $S$ of Aguilera et
al.~\cite{aguilera-podc03-journal}, where in every execution, all channels can
be flaky except those emanating from an a priori unknown correct process
(\emph{timely source}); the latter channels are required to be eventually
reliable and timely. In our framework, $S$ is represented by the set of
executions in $\MCU$ compliant with the following fail-prone system:
\[
  \F_S = \{(P, C) \, \mid\, |P| < n \wedge
  \exists p \in \PP \setminus P.\ \forall q \in \PP \setminus (P \cup \{p\}).\ (p, q) \not\in C\}.
\]

We now use Theorem~\ref{thm:k-consensus-lower} to prove that consensus is
impossible even in a stronger variant of $S$ where only up to a threshold $k$ of
processes are allowed to fail. Formally, for $n$ and $k$ such that $0 < k < n$,
let $\F_{S,k} = \F_S \cap \{(P, C)\, \mid\, |P|\le k\}$. Then $\F_{S,k}$ is a
$k$-fail-prone system.
\begin{theorem}
  Let $n$ and $k$ be such that $0 < k < n$, and 
  $\tau: \F_{S,k} \rightarrow 2^\PP$ be a termination mapping such that
  $\forall f \in \F_{S,k}.\ \tau(f)\neq\emptyset$. Then no algorithm 
  can implement an $(\F_{S,k}, \tau)$-obstruction-free consensus
  in $\MCU$.
\end{theorem}
\begin{proof}
  By contradiction, assume such an algorithm exists.  First,
  a usual partitioning argument similar to that of Dwork et al.~\cite{dls} can
  be used to show that we must have $n>2$ and
  $k \le \lfloor \frac{n-1}{2} \rfloor$.  Then, since $\MCL$ is stronger than
  $\MCU$, Theorem~\ref{thm:k-consensus-lower} requires that for all
  $f \in \F_{S,k}$, $\RG{\GG}{f}$ must contain a strongly connected component
  of size $> \lfloor \frac{n-1}{2} \rfloor \ge 1$. However, for a failure
  pattern $f=(P,C)\in \F_{S,k}$ such that
\[
|P|=\emptyset \wedge
\exists p \in \PP.\ \forall q\in \PP\setminus (P \cup \{p\}).\
\forall r\in\PP\setminus P.\ (q, r) \in C,
\]
all strongly connected components of $\RG{\GG}{f}$ are of size $1$: a
contradiction.
\end{proof}

\subparagraph{Upper bounds.} 
On the positive side, we now show that strengthening connectivity assumptions
does enable consensus to be solved in many previously proposed 
models of weak synchrony, 
such as systems $S^+$ and $S^{++}$ of Aguilera et al.~\cite{aguilera-podc03-journal} and
those defined in~\cite{aguilera-podc04,dahlia-t-accessible,dahila-t-moving-source,antonio-intermittent-star}.
These papers show that $\Omega$ can be implemented in their respective models,
but (with an exception of~\cite{aguilera-podc04}) do not give a consensus algorithm.

To show that consensus is indeed possible in these models, 
we introduce an intermediate model $\MRUH$ in which all channels can be asynchronous
and flaky apart from those connecting an a priori unknown correct process 
(a \emph{hub}~\cite{aguilera-podc03-journal})
to all the other processes in \emph{both} directions; the latter channels are
required to be eventually reliable (but can still be asynchronous). 
In our framework, $\MRUH$ is represented by the set of executions of $\MRU$
compliant with the following fail-prone system:
\[
\F_{\MRUH} = \{(P, C)\, \mid\, |P| < n \wedge
\exists p \in \PP \setminus P.\ \forall q \in \PP \setminus (P \cup \{p\}).\ (p, q) \not\in C \wedge
(q, p) \not\in C\}.
\]
Thus, for all $f\in \F_{\MRUH}$, $\RG{\GG}{f}$ is strongly connected, 
and in the case when at most $\lfloor \frac{n-1}{2} \rfloor$ processes fail in $f$, 
$\RG{\GG}{f}$ is a connected core.
Then from Theorem~\ref{thm:upper-reg}, we get
\begin{corollary}
It is possible to implement a wait-free atomic register in $\MRUH$ if at most
$\lfloor \frac{n-1}{2} \rfloor$ processes can crash.
\label{thm:wait-free-atomic}
\end{corollary} 
Since wait-free consensus 
tolerating any number of $<n$ process crashes can be implemented using atomic
registers and $\Omega$~\cite{LoH94}, Corollary~\ref{thm:wait-free-atomic} implies
\begin{corollary}
It is possible to implement a wait-free consensus in $\MRUH$ augmented
with $\Omega$ if at most
$\lfloor \frac{n-1}{2} \rfloor$ processes can crash.
\label{corol:wait-free-consensus}
\end{corollary}
We now use this result to prove
\begin{corollary}
    \label{cor:upper-reg-hub}
    It is possible to implement wait-free consensus over the systems $S^+$ and $S^{++}$ 
    of~\cite{aguilera-podc03-journal} and the models 
    of~\cite{aguilera-podc04,dahlia-t-accessible,dahila-t-moving-source,antonio-intermittent-star},
    provided at most $\lfloor \frac{n-1}{2} \rfloor$ processes can crash.
\end{corollary}
\begin{proof}
We first show that the systems enumerated in the theorem's statement
are as strong as $\MRUH$.
System $S^+$~\cite{aguilera-podc03-journal} stipulates that every execution has a
\emph{fair hub}, i.e., a correct process which is connected to all other processes
via fair-lossy channels~\cite{lossy-channels-AfekAFFLMWZ94} in both directions. 
Since eventually reliable channels can be built on top of fair-lossy ones in an asynchronous system,
$S^+$ is as strong as $\MRUH$. In $S^{++}$~\cite{aguilera-podc03-journal}
and~\cite{aguilera-podc04}, every pair of processes
is assumed to be connected by fair-lossy channels. Thus, 
similarly to the above, $S^{++}$ is as strong as $\MRUH$. Finally, the models
of~\cite{dahlia-t-accessible,dahila-t-moving-source,antonio-intermittent-star}
assume that every pair of processes is connected via reliable channels, which
again means that these models are strictly stronger than $\MRUH$.  Since $\Omega$ can be
implemented in all models mentioned
above~\cite{aguilera-podc03-journal,aguilera-podc04,dahlia-t-accessible,dahila-t-moving-source,antonio-intermittent-star},
Corollary~\ref{corol:wait-free-consensus} implies the required.
\end{proof}

\section{Related Work}
\label{sec:related}

Fault-tolerant distributed computing in the presence of message loss has been
extensively studied in the past. Early work focused on the models that, in
addition to crashes, allow processes to also fail to either send messages
(\emph{send omission})~\cite{hadzilacos-send-omit} or both send and receive
messages (\emph{generalized omission})~\cite{perry-toueg-omit}.  Both these
models have been shown to be equivalent to the crash failures in their
computational power in both synchronous~\cite{neiger-toueg} and
asynchronous~\cite{coan} systems. They are however too strong to capture
partitionable systems, as they would automatically classify any non-crashed
process with unreliable connectivity as faulty.

The classical paper by Dolev~\cite{dolev-strikes-again} as well as more recent work 
(see~\cite{topology-consensus} for survey) 
study consensus solvability in general networks as a function of the network topology,
process failure models, and synchrony constraints. These papers however, only
consider the classical (i.e., non-partitionable) version of consensus, which
requires agreement to be reached by all correct processes. They also do
not consider failure models with flaky channels.

In a seminal paper~\cite{time-not-healer}, Santoro and Widmayer proposed
a \emph{mobile omission} failure model, which treats message loss
independently of process failures. They showed that in order to solve
consensus in a synchronous round-based system in the absence of process failures, 
it is necessary and sufficient to ensure
that the communication graph contains a strongly
connected component in every round~\cite{time-not-healer,santoro-widmayer2}.
In contrast, our lower bounds demonstrate that 
in order to implement consensus (or even a register)
in a partially synchronous system, some processes must remain
strongly connected throughout the \emph{entire} execution.

Subsequent work explored fault-tolerant computation in the presence of both
faulty links and processes under weakened synchrony assumptions. 
Basu et al.~\cite{basu96} established a separation between 
reliable channels and either eventually reliable or fair-lossy channels 
in terms of their power to solve certain problems 
(such as uniform reliable broadcast and $k$-set agreement) under asynchrony.
However, their notion of a reliable link~\cite{ht-broadcast-tutorial} 
is too strong to be implementable in practice as it 
requires every message transmitted via a complete send invocation to be eventually delivered 
to a correct destination even if the sender later crashes.

Several generalizations of the classical failure detector abstractions of~\cite{CT96}
for partitionable systems
were proposed in~\cite{friedman1,friedman2,friemdan-podc-ba,aguilera-heartbeat}
and shown sufficient for consensus in~\cite{friedman1,friedman2,aguilera-heartbeat}. 
However, as we explain in \S\ref{sec:intro}, the failure detectors introduced
in these papers cannot be implemented in the presence of flaky channels.

Aguilera et al.~\cite{aguilera-podc03-journal} and follow-up 
work~\cite{aguilera-podc04,dahlia-t-accessible,dahila-t-moving-source,antonio-intermittent-star}
studied the problem of implementing a failure detector
$\Omega$ (which is the weakest for consensus~\cite{CT96-weakest}),
under various weak models of synchrony, link reliability, and connectivity.
However, these results are inapplicable in our setting as
the models considered in these papers disallow full partitions.

As we show in~\S\ref{sec:cons-via-omega}, our connectivity bounds for
consensus are also applicable in non-partitionable systems with
unreliable channels. In particular, they imply that $\Omega$ is not the only
factor determining consensus solvability, and that the degree of
connectivity being assumed also plays an important role. 
This suggests an interesting tradeoff across all three modeling dimensions 
of synchrony, channel reliability, and connectivity considered in this paper. 
In particular, to solve consensus we assume that every non-faulty channel 
is eventually timely~\cite{dls}, whereas the algorithm of~\cite{aguilera-podc04}
can cope with weaker synchrony constraints.
However, unlike~\cite{aguilera-podc04}, 
which assumes complete connectivity via fair-lossy channels, 
our implementation is able to tolerate an adversarial message loss (flakiness).
Whether this tradeoff is fundamental, or our synchrony assumptions
can be further relaxed 
remains the subject of future work.

Alquraan et al.~\cite{osdi-partitions} present a study of system failures due to
network partitions, which we already mentioned in \S\ref{sec:intro}. This work
highlights the practical importance of coming up with provably correct designs
that explicitly consider channel failures. A follow-up
paper~\cite{osdi-partitions2} introduces a communication layer that can mask
channel failures, but only when access to low-level networking infrastructure is
available. OmniPaxos~\cite{omnipaxos} ensures liveness of consensus under
channel failures, but only handles disconnected channels, not flaky ones.

Raft~\cite{raft,raft-thesis} elects leaders in a randomized way: a prospective
leader picks a view number and requests {\em Votes} from a majority of processes
(analogous to $\mathsf{1A}$ messages of Paxos). This may lead to split votes when multiple
processes are competing to get elected, in which case each process waits for a
randomized timeout and retries. While this mechanism works reasonably well
in practice, it does not guarantee termination under partial synchrony even
with perfect connectivity. Raft also includes a {\em Pre-Vote} optimization,
which, as Jensen et al. point out~\cite{heidi-raft}, can help maintain
liveness under intermittent connectivity. The optimization
requires that before sending {\em Vote} requests, a prospective leader gathers a
majority of {\em Pre-Votes} from other processes, certifying that they are ready
to vote for it. This is similar to gathering $\WISH$es from a majority of
processes in our synchronizer before entering a view
(Figure~\ref{fig:synchronizer_implementation}). But unlike our synchronizer {\em
  Pre-Vote} is a best-effort technique, since in between a process granting a
{\em Pre-Vote} and receiving the corresponding {\em Vote} request it may vote
for someone else, thereby creating a split vote. Nevertheless, the presence of
techniques similar to ours in existing practical algorithms makes us hopeful
that our work could be used to make Paxos-based systems provably live even under
adversarial network conditions.

\bibliographystyle{plainurl}
\bibliography{references.bib}

\begin{thebibliography}{10}

\bibitem{adya99-weak-consis}
Atul Adya.
\newblock {\em Weak Consistency: A Generalized Theory and Optimistic
  Implementations for Distributed Transactions}.
\newblock {Ph.D.}, MIT, Cambridge, MA, USA, March 1999.
\newblock Also as Technical Report MIT/LCS/TR-786.

\bibitem{lossy-channels-AfekAFFLMWZ94}
Yehuda Afek, Hagit Attiya, Alan~D. Fekete, Michael~J. Fischer, Nancy~A. Lynch,
  Yishay Mansour, Da{-}Wei Wang, and Lenore~D. Zuck.
\newblock Reliable communication over unreliable channels.
\newblock {\em J. {ACM}}, 41(6):1267--1297, 1994.

\bibitem{aguilera-podc04}
Marcos~K. Aguilera, Carole Delporte-Gallet, Hugues Fauconnier, and Sam Toueg.
\newblock Communication-efficient leader election and consensus with limited
  link synchrony.
\newblock In {\em Symposium on Principles of Distributed Computing (PODC)},
  2004.

\bibitem{aguilera-podc03-journal}
Marcos~K. Aguilera, Carole Delporte-Gallet, Hugues Fauconnier, and Sam Toueg.
\newblock On implementing {O}mega in systems with weak reliability and
  synchrony assumptions.
\newblock {\em Distributed Comput.}, 21(4):285--314, 2008.

\bibitem{aguilera-heartbeat}
Marcos~Kawazoe Aguilera, Wei Chen, and Sam Toueg.
\newblock Using the heartbeat failure detector for quiescent reliable
  communication and consensus in partitionable networks.
\newblock {\em Theor. Comput. Sci.}, 220(1):3--30, 1999.

\bibitem{crash-rec-consensus-AguileraCT00}
Marcos~Kawazoe Aguilera, Wei Chen, and Sam Toueg.
\newblock Failure detection and consensus in the crash-recovery model.
\newblock {\em Distributed Comput.}, 13(2):99--125, 2000.

\bibitem{osdi-partitions2}
Mohammed Alfatafta, Basil Alkhatib, Ahmed Alquraan, and Samer Al-Kiswany.
\newblock Toward a generic fault tolerance technique for partial network
  partitioning.
\newblock In {\em Symposium on Operating Systems Design and Implementation
  (OSDI)}, 2020.

\bibitem{osdi-partitions}
Ahmed Alquraan, Hatem Takruri, Mohammed Alfatafta, and Samer Al-Kiswany.
\newblock An analysis of network-partitioning failures in cloud systems.
\newblock In {\em Symposium on Operating Systems Design and Implementation
  (OSDI)}, 2018.

\bibitem{abd}
Hagit Attiya, Amotz Bar{-}Noy, and Danny Dolev.
\newblock Sharing memory robustly in message-passing systems.
\newblock {\em J. {ACM}}, 42(1):124--142, 1995.

\bibitem{obf-formal}
Hagit Attiya, Ohad Ben-Baruch, and Danny Hendler.
\newblock Lower bound on the step complexity of anonymous binary consensus.
\newblock In {\em Symposium on Distributed Computing (DISC)}, 2016.

\bibitem{bailis-kingsbury}
Peter Bailis and Kyle Kingsbury.
\newblock The network is reliable.
\newblock {\em Commun. {ACM}}, 57(9):48--55, 2014.

\bibitem{basu96}
Anindya Basu, Bernadette Charron-Bost, and Sam Toueg.
\newblock Simulating reliable links with unreliable links in the presence of
  process crashes.
\newblock In {\em Workshop on Distributed Algorithms (WDAG)}, 1996.

\bibitem{multishot-disc22}
Manuel Bravo, Gregory Chockler, and Alexey Gotsman.
\newblock Liveness and latency of {B}yzantine state-machine replication.
\newblock In {\em Symposium on Distributed Computing (DISC)}, 2022.

\bibitem{bftlive-dc}
Manuel Bravo, Gregory Chockler, and Alexey Gotsman.
\newblock Making {B}yzantine consensus live.
\newblock {\em Distributed Comput.}, 35(6):503--532, 2022.

\bibitem{brewer}
Eric~A. Brewer.
\newblock Towards robust distributed systems (abstract).
\newblock In {\em Symposium on Principles of Distributed Computing (PODC)},
  2000.

\bibitem{physalia}
Marc Brooker, Tao Chen, and Fan Ping.
\newblock Millions of tiny databases.
\newblock In {\em Symposium on Networked Systems Design and Implementation
  (NSDI)}, 2020.

\bibitem{CT96-weakest}
Tushar~Deepak Chandra, Vassos Hadzilacos, and Sam Toueg.
\newblock The weakest failure detector for solving consensus.
\newblock {\em J. ACM}, 43(4):685--722, 1996.

\bibitem{CT96}
Tushar~Deepak Chandra and Sam Toueg.
\newblock Unreliable failure detectors for reliable distributed systems.
\newblock {\em J. ACM}, 43(2):225–267, 1996.

\bibitem{gcs-survey}
Gregory Chockler, Idit Keidar, and Roman Vitenberg.
\newblock Group communication specifications: A comprehensive study.
\newblock {\em {ACM} Comput. Surv.}, 33(4):427--469, 2001.

\bibitem{coan}
Brian Coan.
\newblock A compiler that increases the fault tolerance of asynchronous
  protocols.
\newblock {\em IEEE Trans. Comput.}, 37(12):1541–1553, 1988.

\bibitem{dolev-strikes-again}
Danny Dolev.
\newblock The {B}yzantine generals strike again.
\newblock {\em J. Algorithms}, 3(1):14--30, 1982.

\bibitem{friedman2}
Danny Dolev, Roy Friedman, Idit Keidar, and Dahlia Malkhi.
\newblock {Failure detectors in omission failure environments}.
\newblock Technical Report TR96-1608, Department of Computer Science, Cornell
  University, 1996.

\bibitem{friemdan-podc-ba}
Danny Dolev, Roy Friedman, Idit Keidar, and Dahlia Malkhi.
\newblock Failure detectors in omission failure environments (brief
  announcement).
\newblock In {\em Symposium on Principles of Distributed Computing (PODC)},
  1997.

\bibitem{dls}
Cynthia Dwork, Nancy~A. Lynch, and Larry~J. Stockmeyer.
\newblock Consensus in the presence of partial synchrony.
\newblock {\em J. {ACM}}, 35(2):288--323, 1988.

\bibitem{antonio-intermittent-star}
Antonio Fern{\'a}ndez~Anta and Michel Raynal.
\newblock From an intermittent rotating star to a leader.
\newblock In {\em Conference on Principles of Distributed Systems (OPODIS)},
  2007.

\bibitem{solo-orig}
Faith Fich, Maurice Herlihy, and Nir Shavit.
\newblock On the space complexity of randomized synchronization.
\newblock {\em J. ACM}, 45(5):843–862, 1998.

\bibitem{friedman1}
Roy Friedman, Idit Keidar, Dahlia Malkhi, Ken Birman, and Danny Dolev.
\newblock {Deciding in partitionable networks}.
\newblock Technical Report TR95-1554, Department of Computer Science, Cornell
  University, 1995.

\bibitem{cap}
Seth Gilbert and Nancy Lynch.
\newblock Brewer's conjecture and the feasibility of consistent, available,
  partition-tolerant web services.
\newblock {\em SIGACT News}, 33(2):51–59, 2002.

\bibitem{hadzilacos-send-omit}
Vassos Hadzilacos.
\newblock {Byzantine agreement under restricted type of failures (not telling
  the truth is different from telling lies)}.
\newblock Technical Report TR-18-63, Department of Computer Science, Harvard
  University, 1983.

\bibitem{ht-broadcast-tutorial}
Vassos Hadzilacos and Sam Toueg.
\newblock {A modular approach to fault-tolerant broadcast and related
  problems}.
\newblock Technical Report TR94-1425, Department of Computer Science, Cornell
  University, 1994.

\bibitem{linear}
Maurice Herlihy.
\newblock Wait-free synchronization.
\newblock {\em ACM Transactions on Programming Languages and Systems},
  13(1):124--149, 1991.

\bibitem{obf-orig}
Maurice Herlihy, Victor Luchangco, and Mark Moir.
\newblock Obstruction-free synchronization: Double-ended queues as an example.
\newblock In {\em International Conference on Distributed Computing Systems
  (ICDCS)}, 2003.

\bibitem{dahila-t-moving-source}
Martin Hutle, Dahlia Malkhi, Ulrich Schmid, and Lidong Zhou.
\newblock Chasing the weakest system model for implementing $\omega$ and
  consensus.
\newblock {\em IEEE Trans. Dependable Secur. Comput.}, 6(4):269--281, 2009.

\bibitem{heidi-raft}
Chris Jensen, Heidi Howard, and Richard Mortier.
\newblock Examining {R}aft's behaviour during partial network failures.
\newblock In {\em Workshop on High Availability and Observability of Cloud
  Systems (HAOC)}, 2021.

\bibitem{lamport1985on}
Leslie Lamport.
\newblock On interprocess communication - part {I}: Basic formalism, part {II}:
  Algorithms.
\newblock {\em Distributed Comput.}, 1(2):77--101, 1986.

\bibitem{paxos}
Leslie Lamport.
\newblock The part-time parliament.
\newblock {\em ACM Trans. Comput. Syst.}, 16(2):133–169, 1998.

\bibitem{LoH94}
Wai{-}Kau Lo and Vassos Hadzilacos.
\newblock Using failure detectors to solve consensus in asynchronous
  shared-memory systems (extended abstract).
\newblock In {\em Workshop on Distributed Algorithms (WDAG)}, 1994.

\bibitem{lynch-dist-algos}
Nancy Lynch.
\newblock {\em Distributed Algorithms}, chapter~17.
\newblock Morgan Kaufmann, 1996.

\bibitem{dahlia-t-accessible}
Dahlia Malkhi, Florin Oprea, and Lidong Zhou.
\newblock $\omega$ meets paxos: Leader election and stability without eventual
  timely links.
\newblock In {\em Symposium on Distributed Computing (DISC)}, 2005.

\bibitem{bqs}
Dahlia Malkhi and Michael~K. Reiter.
\newblock Byzantine quorum systems.
\newblock {\em Distributed Comput.}, 11(4):203--213, 1998.

\bibitem{cogsworth}
Oded Naor, Mathieu Baudet, Dahlia Malkhi, and Alexander Spiegelman.
\newblock Cogsworth: {B}yzantine view synchronization.
\newblock In {\em Cryptoeconomics Systems Conference (CES)}, 2020.

\bibitem{oded-linear}
Oded Naor and Idit Keidar.
\newblock Expected linear round synchronization: The missing link for linear
  {B}yzantine {SMR}.
\newblock In {\em Symposium on Distributed Computing (DISC)}, 2020.

\bibitem{neiger-toueg}
Gil Neiger and Sam Toueg.
\newblock Automatically increasing the fault-tolerance of distributed
  algorithms.
\newblock {\em J. Algorithms}, 11(3):374–419, 1990.

\bibitem{omnipaxos}
Harald Ng, Seif Haridi, and Paris Carbone.
\newblock {Omni-Paxos}: Breaking the barriers of partial connectivity.
\newblock In {\em European Conference on Computer Systems (EuroSys)}, 2023.

\bibitem{raft-thesis}
Diego Ongaro.
\newblock {\em Consensus: bridging theory and practice}.
\newblock PhD thesis, Stanford University, {USA}, 2014.

\bibitem{raft}
Diego Ongaro and John~K. Ousterhout.
\newblock In search of an understandable consensus algorithm.
\newblock In {\em USENIX Annual Technical Conference}, 2014.

\bibitem{perry-toueg-omit}
Kenneth~J. Perry and Sam Toueg.
\newblock Distributed agreement in the presence of processor and communication
  faults.
\newblock {\em IEEE Trans. Software Eng.}, 12(3):477--482, 1986.

\bibitem{paxos-with-omissions}
Roberto~De Prisco, Butler~W. Lampson, and Nancy~A. Lynch.
\newblock Revisiting the {PAXOS} algorithm.
\newblock {\em Theor. Comput. Sci.}, 243(1-2):35--91, 2000.

\bibitem{topology-consensus}
Dimitris Sakavalas and Lewis Tseng.
\newblock {\em Network Topology and Fault-Tolerant Consensus}.
\newblock Synthesis Lectures on Distributed Computing Theory. Morgan {\&}
  Claypool Publishers, 2019.

\bibitem{time-not-healer}
Nicola Santoro and Peter Widmayer.
\newblock Time is not a healer.
\newblock In {\em Symposium on Theoretical Aspects of Computer Science
  (STACS)}, 1989.

\bibitem{santoro-widmayer2}
Nicola Santoro and Peter Widmayer.
\newblock Distributed function evaluation in the presence of transmission
  faults.
\newblock In {\em Symposium on Algorithms (SIGAL)}, 1990.

\end{thebibliography}

\iflong
    \appendix
    \newpage
    \section{Proof of Theorem~\ref{thm:upper-reg}}
\label{sec:app-abd}

\begin{figure}[h]
\scalebox{0.97}{
    \begin{tabular}{@{\qquad}l@{}|@{\qquad}l@{}}
    \hspace{-28px}
    \begin{minipage}[t]{7.8cm}
    \vspace{-13px}
    \begin{algorithm}[H]
        \SetAlgoNoLine
        \DontPrintSemicolon
        \setcounter{AlgoLine}{0}
        \SetInd{0.5em}{0.5em}

        \Function{$\oread()$}{\label{abd:read_req}
            \precond $\status = \idle$\;
            $\assign{\seq}{\seq + 1}$\;
            $\assign{\query[i]}{\seq}$\;
            $\assign{\status}{\rdquery}$\;
            \wait $\status = \rddone$\;
            $\assign{\status}{\idle}$\;
            \return $\rdval$\;
        }

        \smallskip
        \smallskip

        \Function{$\owrite(v)$}{\label{abd:write_req}
            \precond $\status = \idle$\;
            $\assign{\wrval}{v}$\;
            $\assign{\seq}{\seq + 1}$\;
            $\assign{\query[i]}{\seq}$\;
            $\assign{\status}{\wrquery}$\;
            \wait $\status = \wrdone$\;
            $\assign{\status}{\idle}$\;
            \return $\ACK$\;
        }

        \smallskip
        \smallskip

        \SubAlgo{${\bf periodically}$}{\label{abd:flood}
            \send $\STATE(\query,\queryack,\linebreak
            \wwrite,\wwriteack)$ \ToAll\;
        }

        \smallskip
        \smallskip

        \SubAlgo{\onreceive $\STATE(Q,R,W,X)$}{\label{abd:update_state}
            \For{$p_j \in \mathcal{P}$}{
                \If{$Q[j].\sseq > \query[j].\sseq$}{\label{abd:update_state:query}
                    $\assign{\query[j]}{Q[j]}$\;
                }
                \If{$W[j].\sseq > \wwrite[j].\sseq$}{\label{abd:update_state:write}
                    $\assign{\wwrite[j]}{W[j]}$\;
                }
                \For{$p_k \in \mathcal{P}$}{
                    \If{$R[j][k].\sseq > \queryack[j][k].\sseq$}{\label{abd:update_state:queryack}
                        $\assign{\queryack[j][k]}{R[j][k]}$\;
                    }
                    \If{$X[j][k].\sseq > \wwriteack[j][k].\sseq$}{\label{abd:update_state:writeack}
                        $\assign{\wwriteack[j][k]}{X[j][k]}$\;
                    }
                }
            }
        }

        \smallskip
        \smallskip
    \end{algorithm}
    \end{minipage}
&
    \hspace{-15px}
    \begin{minipage}[t]{7.1cm}
    \vspace{-13px}
    \begin{algorithm}[H]
        \SetAlgoNoLine
        \DontPrintSemicolon
        \SetInd{0.5em}{0.5em}

        \SubAlgo{\when $\exists j.\ \query[j].\sseq > \queryack[i][j].\sseq$}{\label{abd:query}
            $\assign{\queryack[i][j]}{(\val,\ts,\query[j].\sseq)}$\;\label{abd:query:ack}
        }

        \smallskip
        \smallskip

        \SubAlgo{\when $\exists j.\ \wwrite[j].\sseq > \wwriteack[i][j].\sseq$}{\label{abd:write}
            \If{$\wwrite[j].\sts > \ts$}{\label{abd:write:guard}
                $\assign{\val}{\wwrite[j].\sval}$\;
                $\assign{\ts}{\wwrite[j].\sts}$\;\label{abd:write:ts}
            }
            $\assign{\wwriteack[i][j]}{\wwrite[j].\sseq}$\;
        }

        \smallskip
        \smallskip

        \SubAlgo{\when $\status = \rdquery \land \linebreak
            \phantom\ \ \ |\{j\ |\ \queryack[j][i].\sseq = \seq\}| > \frac{n}{2}$}{\label{abd:read_quorum}
            $\assign{Q}{\{j\ |\ \queryack[j][i].\sseq = \seq\}}$\;\label{abd:read_quorum:Q}
            $\assign{j}{\mathsf{argmax}\{\queryack[j][i].\sts\ |\ j \in Q\}}$\;\label{abd:rdquery:max}
            $\assign{\rdval}{\queryack[j][i].\sval}$\;\label{abd:rdquery:rdval}
            $\assign{\val}{\queryack[j][i].\sval}$\;
            $\assign{\ts}{\queryack[j][i].\sts}$\;\label{abd:rdquery:ts}
            $\assign{\status}{\rdpropagate}$\;
            $\assign{\wwrite[i]}{(\val,\ts,\seq)}$\;
        }

        \smallskip
        \smallskip

        \SubAlgo{\when $\status =\wrquery \land \linebreak
            \phantom\ \ \ |\{j\ |\ \queryack[j][i].\sseq = \seq\}| > \frac{n}{2}$}{\label{abd:write_quorum}
            $\assign{Q}{\{j\ |\ \queryack[j][i].\sseq = \seq\}}$\;\label{abd:write_quorum:Q}
            $\assign{(t,\_)}{\max\{\queryack[j][i].\sts\ |\ j \in Q\}}$\;\label{abd:wrquery:max}
            $\assign{\val}{\wrval}$\;
            $\assign{\ts}{(t+1,i)}$\;\label{abd:wrquery:ts}
            $\assign{\status}{\wrpropagate}$\;
            $\assign{\wwrite[i]}{(\val,\ts,\seq)}$\;
        }

        \smallskip
        \smallskip

        \SubAlgo{\when $\status = \rdpropagate \land \linebreak
            \phantom\ \ \ |\{j\ |\ \wwriteack[j][i].\sseq = \seq\}| > \frac{n}{2}$}{\label{abd:read_complete}
            $\assign{\status}{\rddone}$\;
        }

        \smallskip
        \smallskip

        \SubAlgo{\when $\status = \wrpropagate \land \linebreak
            \phantom\ \ \ |\{j\ |\ \wwriteack[j][i].\sseq = \seq\}| > \frac{n}{2}$}{\label{abd:write_complete}
            $\assign{\status}{\wrdone}$\;
        }
    \end{algorithm}
    \end{minipage}
    \end{tabular}
}
    \caption{Register protocol at a process $p_i$.}
    \label{fig:abd_protocol}
\end{figure}

In Figure~\ref{fig:abd_protocol} we present an implementation of an
atomic register protocol in the model $\MRU$ that validates
Theorem~\ref{thm:upper-reg}. The protocol is a variant of ABD~\cite{abd}.
We assume that for every process the initial values of $\val$, $\ts$, $\seq$ and $\status$ are
$0$, $(0,0)$, $0$ and $\idle$, respectively, and that $\ts$s are compared lexicographically.

\newcommand{\WR}{{\sf wr}}
\newcommand{\WW}{{\sf ww}}
\newcommand{\RW}{{\sf rw}}
\newcommand{\RT}{{\sf rt}}

We first show that the algorithm is linearizable.
For simplicity, we only consider executions of the algorithm where all
operations complete.
To each execution
$\sigma$ of the algorithm in Figure~\ref{fig:abd_protocol}, we associate:
\begin{itemize}
    \item a set $V(\sigma)$ consisting of the operations in
      $\sigma$, i.e., $\oread$s and $\owrite$s; and
    \item a relation $\RT(\sigma)$, defined as follows:
      for all $o_1,o_2\in\sigma$, $(o_1,o_2)\in\RT$
    if and only if $o_1$ completes before $o_2$ is invoked.
\end{itemize}

We denote the read operations in $\sigma$ by $R(\sigma)$ and the write operations in $\sigma$ by $W(\sigma)$. A \emph{dependency graph} of $\sigma$ is a tuple
$G=(V(\sigma), \RT(\sigma), \WR, \WW, \RW)$, where the relations
$\WR, \WW, \RW \subseteq V(\sigma) \times V(\sigma)$ are such that:
\begin{enumerate}
\item $(i)$ if $(o_1, o_2) \in \WR$, then $o_1 \in W(\sigma)$ and $o_2 \in R(\sigma)$;\\
        $(ii)$ for all $w_1,w_2,r\in\sigma$ such that $(w_1,r)\in\WR$ and $(w_2,r)\in\WR$, we have $w_1=w_2$;\\
  $(iii)$ for all $(w,r)\in\WR$ we have $\val(w)=\val(r)$;\\
        $(iv)$ if there is no $w \in W(\sigma)$ such that $(w,r)\in\WR$, then $r$
        returns $0$;
    \item $\WW$ is a total order over $W(\sigma)$; and
    \item $\RW=\{(r,w)\ |\ \exists w'.\ (w',r)\in\WR \land (w',w)\in\WW\} \cup\\
        \phantom\ \ \ \ \ \ \ \{(r,w)\ |\ r\in R(\sigma)\land w\in W(\sigma)\land \neg\exists w'.\ (w',r)\in\WR\}$.
\end{enumerate}

To prove linearizability we rely on the following theorem~\cite{adya99-weak-consis}:
\begin{theorem}
    \label{thm:linearizability}
    An execution $\sigma$ is linearizable if and only if there exist $\WR$, $\WW$ and $\RW$
    such that $G=(V(\sigma), \RT(\sigma), \WR, \WW, \RW)$ is an acyclic dependency graph.
\end{theorem}

We now prove that every execution of the algorithm in Figure~\ref{fig:abd_protocol}
is linearizable. Fix one such execution $\sigma$. Our strategy is to find witnesses for
$\WR$, $\WW$ and $\RW$ that validate the conditions of Theorem~\ref{thm:linearizability}.
For this, consider the function $\tau : \sigma \rightarrow \mathbb{N}\times\mathbb{N}$
that maps each operation in $\sigma$ to a tag as follows:
\begin{itemize}
    \item for a read $r$, $\tau(r)$ is the tag set at line~\ref{abd:rdquery:ts}; and
    \item for a write $w$, $\tau(w)$ is the tag set at line~\ref{abd:wrquery:ts}.
\end{itemize}

We then define the required witnesses as follows:
\begin{itemize}
    \item $(w,r)\in\WR$ if and only if $w \in W(\sigma)$, $r \in R(\sigma)$
        and $\tau(w)=\tau(r)$;
    \item $(w,w')\in\WW$ if and only if $w,w' \in W(\sigma)$
      and $\tau(w)<\tau(w')$; and
    \item $\RW$ is derived from $\WR$ and $\WW$ as per the dependency graph definition.
\end{itemize}

Our next results rely on the following proposition, whose easy proof we omit.
\begin{proposition}
  \label{abd:tags_prop}
  \
    \begin{enumerate}
        \item For every $w_1,w_2\in W(\sigma)$, $\tau(w_1)=\tau(w_2)$ implies $w_1=w_2$.
        \item For every $w\in W(\sigma)$, $\tau(w)>(0,0)$.\label{abd:tags_prop:3}
        \item For every $r\in R(\sigma)$, either $\tau(r)=(0,0)$ or there
        exists $w\in W(\sigma)$ such that $\tau(r)=\tau(w)$.\label{abd:tags_prop:1}
        \item For every $r\in R(\sigma)$ and $w\in W(\sigma)$,
            $\tau(r)=\tau(w)$ implies $\val(r)=\val(w)$.
    \end{enumerate}
\end{proposition}

\begin{lemma}
  \label{abd:tau-non-inc}
  \
    \begin{enumerate}
        \item For all $r,w\in V(\sigma)$, if $(r,w)\in\RW$ then $\tau(r)<\tau(w)$.
        \item For all $o_1,o_2\in V(\sigma)$, if $(o_1,o_2)\in\RT$, then $\tau(o_1)\leq\tau(o_2)$.
            Moreover, if $o_2$ is a $\owrite$, then $\tau(o_1)<\tau(o_2)$.
    \end{enumerate}
\end{lemma}
\begin{proof}
  \
    \begin{enumerate}
    \item Let $r,w\in V(\sigma)$ be such that $(r,w)\in\RW$. We consider two cases.
        \begin{itemize}
            \item Suppose that for some
                $w'$ we have $(w',r)\in\WR$ and $(w',w)\in\WW$. The
                definition of $\WR$ implies that $\tau(r)=\tau(w')$, and the definition of
                $\WW$ implies that $\tau(w')<\tau(w)$. Then $\tau(r)<\tau(w)$.
            \item Suppose now that $\neg \exists w'.\ (w',r)\in\WR$.
                We show that $\tau(r)=(0,0)$.
                Indeed, if $\tau(r)\neq(0,0)$, then by Proposition~\ref{abd:tags_prop}.\ref{abd:tags_prop:1}, there exists $w\in W(\sigma)$
                such that $\tau(r)=\tau(w)$. But then $(w,r)\in\WR$, contradicting the
                assumption that there is no such write.
                At the same time, Proposition~\ref{abd:tags_prop}.\ref{abd:tags_prop:3} implies that
                $\tau(w)>(0,0)$. Then $\tau(r)<\tau(w)$.
        \end{itemize}
        \item Let $o_1,o_2\in V(\sigma)$ be such that $(o_1,o_2)\in\RT$. The
          handlers at lines~\ref{abd:read_complete}
          and~\ref{abd:write_complete} ensure that $o_1$ completes only
          if a majority of processes acknowledged $o_1$'s propagation phase, and
          by the guard at line~\ref{abd:write:guard}, have $\ts \geq \tau(o_1)$.
          This, together with the fact that $\ts$ is non-decreasing at all processes,
          ensures that, after $o_1$
          completes, a majority of processes have $\ts \geq \tau(o_1)$.
            Suppose that $o_2$ is invoked by a process $p_i$.
            Because $o_2$ is invoked after $o_1$ completes, the guards at
            lines~\ref{abd:read_quorum} and~\ref{abd:write_quorum} ensure that for at least one
            $j \in Q$ (lines~\ref{abd:read_quorum:Q}
            and~\ref{abd:write_quorum:Q}) we have
            $\queryack[j][i].\sts \geq \tau(o_1)$. There are two cases:
            \begin{itemize}
            \item Suppose $o_2$ is a read.  Because the $\ts$ associated to
              $o_2$ is the maximum among those corresponding to processes in $Q$
              (lines~\ref{abd:rdquery:max} and~\ref{abd:rdquery:ts}), we have
              $\tau(o_1)\leq\tau(o_2)$.
            \item Suppose $o_2$ is a write.  Because the $\ts$ associated to
              $o_2$ is greater than the maximum among those corresponding to
              processes in $Q$ (lines~\ref{abd:wrquery:max}
              and~\ref{abd:wrquery:ts}), we have $\tau(o_1)<\tau(o_2)$.
        \end{itemize}
    \end{enumerate}
\end{proof}

\begin{lemma}
    \label{thm:dep-graph-acyc}
    $G=(V(\sigma), \RT(\sigma), \WR, \WW, \RW)$ is an acyclic dependency graph.
\end{lemma}
\begin{proof}
    From Proposition~\ref{abd:tags_prop} and the
    definitions of $\WR$, $\WW$ and $\RW$ it easily follows
    that $G$ is a dependency graph.
    We now show that $G$ is acyclic.
    By contradiction, assume that the graph contains a cycle $o_1, \dots, o_n = o_1$.
    Then $n > 1$. By Lemma~\ref{abd:tau-non-inc} and the definitions of $\tau$
    and $\WW$, we must have $\tau(o_1) \leq \dots \leq \tau(o_n) = \tau(o_1)$,
    so that $\tau(o_1) = \dots = \tau(o_n)$. Furthermore,
    if $(o,o')$ is an edge of $G$ and $o'$ is a write, then $\tau(o)<\tau(o')$.
    Hence, all the operations in the cycle must be reads, and thus, all the edges
    in the cycle come from $\RT$. Then there exist reads $r_1,r_2$
    in the cycle such that $r_1$ completes before $r_2$ is invoked and $r_2$
    completes before $r_1$ is invoked, which is a contradiction.
\end{proof}

Theorem~\ref{thm:upper-reg} follows from Theorem~\ref{thm:linearizability},
Lemma~\ref{thm:dep-graph-acyc} and the following
lemma showing that the algorithm is live.

\begin{lemma}
    Let $\FS$ be a fail-prone system such that for all
    $f \in \FS$, the graph $\RG{\GG}{f}$ contains a
    connected core $\CQ_f$, and let $\tau: \FS \rightarrow 2^\PP$ be the
    termination mapping such that $\forall f \in \FS.\ \tau(f) = \CQ_f$.
    Then the algorithm in Figure~\ref{fig:abd_protocol} is
    $(\FS,\tau)$-wait-free in the model $\MRU$.
\end{lemma}
\begin{proof}
Suppose a process $p_i \in \CQ_f$ invokes a $\oread$
operation (line~\ref{abd:read_req}). Then $p_i$ increases its $\seq$
number, updates its $\query$ array and sets its $\status$ to $\rdquery$.
Let $p_j\in\CQ_f$. Then $p_i$ and $p_j$ are strongly connected by
correct channels and the gossiping mechanism at lines~\ref{abd:flood} and~\ref{abd:update_state}
ensure that if $p_i$ or $p_j$ propagate its information then it is eventually received by
the other.

Because $p_i$ periodically
propagates its $\query$ array (line~\ref{abd:flood}), then $p_j$ eventually
receives it (line~\ref{abd:update_state}).
Therefore, the guard at line~\ref{abd:update_state:query} is validated
and $p_j$ incorporates $p_i$'s $\query$ into its local copy.
Since the $\seq$ numbers received by $p_j$ from $p_i$ are increasing,
the guard at line~\ref{abd:query} is validated. Then
$p_j$ updates $\queryack[j][i]$ with its $\val$, $\ts$ and the corresponding
$\seq$ number.

Because $p_j$ periodically
propagates its $\queryack$ array (line~\ref{abd:flood}), then $p_i$ eventually
receives it (line~\ref{abd:update_state}).
Therefore, the guard at line~\ref{abd:update_state:queryack} is validated
and $p_i$ incorporates $p_j$'s $\queryack$ into its local copy.
Since $p_j$ was picked arbitrarily, $p_i$ does not increment $\seq$ until
the current $\oread$ completes and $|\CQ_f|>\frac{n}{2}$, the guard at
line~\ref{abd:read_quorum} is validated.
Then $p_i$ updates $\status$ to $\rdpropagate$ and $\wwrite[i]$
with its $\val$, $\ts$ and the corresponding
$\seq$ number.

Because $p_i$ periodically
propagates its $\wwrite$ array (line~\ref{abd:flood}), then $p_j$ eventually
receives it (line~\ref{abd:update_state}).
Therefore, the guard at line~\ref{abd:update_state:write} is validated
and $p_j$ incorporates $p_i$'s $\wwrite$ into its local copy.
Since the $\seq$ numbers received by $p_j$ from $p_i$ are increasing,
then the guard at line~\ref{abd:write} is validated.
Then $p_j$ updates $\wwriteack[j][i]$ with the corresponding
$\seq$ number.

Because $p_j$ periodically
propagates its $\wwriteack$ array (line~\ref{abd:flood}), then $p_i$ eventually
receives it (line~\ref{abd:update_state}).
Therefore, the guard at line~\ref{abd:update_state:writeack} is validated
and $p_i$ incorporates $p_j$'s $\wwriteack$ into its local copy.
Since $p_j$ was picked arbitrarily, $p_i$ does not increment $\seq$ until
the current $\oread$ completes and $|\CQ_f|>\frac{n}{2}$, the guard at
line~\ref{abd:read_complete} is validated. Thus the $\oread$
completes.

The case for a $\owrite$ is similar and we omit it from the proof.
\end{proof}

    \section{Details of the Proofs of Theorems~\ref{thm:consensus-lower-1}-\ref{thm:k-consensus-lower}}
\label{sec:app-reg-from-consensus}

\newcommand{\wwait}{{\sf wait}}
\newcommand{\idx}{{\sf idx}}
\newcommand{\res}{{\sf res}}
\newcommand{\op}{{\sf op}}

\newcommand{\execute}{\mathsf{execute}}

\newcommand{\false}{{\tt FALSE}}
\newcommand{\true}{{\tt TRUE}}
\newcommand{\opp}{\mathit{op}}

As we explained in Section~\ref{sec:consensus-lower},
Theorems~\ref{thm:consensus-lower-1} and~\ref{thm:consensus-lower-2}
follow from the fact that registers can be implemented
from consensus. This is well-known for wait-freedom in the classical crash-stop
model with reliable channels. In this section we
show that this result also holds for obstruction-freedom in our setting.
We do this via the algorithm in Figure~\ref{fig:reg-from-consensus}, where
processes use an array of obstruction-free consensus instances to agree on a
sequence of operations executed on the register.

\begin{figure}[h]
    \begin{tabular}{@{\quad\ \ }l@{}|@{\quad \ \ }l@{}}
    \hspace*{-23pt}
    \scalebox{0.96}{%
        \begin{minipage}[t]{7cm}
        \vspace*{-12pt}
        \begin{algorithm}[H]
            \SetAlgoNoLine
            \DontPrintSemicolon
            \setcounter{AlgoLine}{0}
            \SetInd{0.4em}{0.4em}

            $\assign{\idx}{0}$\;
            $\assign{\val}{\bot}$\;
            \text{let $C$ be an infinite array of}\linebreak
                \text{$(\FS,\tau)$-obstruction-free consensus instances}\;

            \SubAlgo{{\bf upon invocation of} $\mathit{\opp}$}{
                \precond $\op=\bot$\;
                $\assign{\op}{\opp}$\;
                $C[\idx].\propose(\op)$\;
            }
        \end{algorithm}
        \vspace*{-5pt}
        \end{minipage}
    }
&
    \scalebox{0.96}{%
        \begin{minipage}[t]{7.5cm}
        \vspace*{-12pt}
        \begin{algorithm}[H]
            \SetAlgoNoLine
            \DontPrintSemicolon
            \SetInd{0.4em}{0.4em}

            \SubAlgo{{\bf upon }$C[\idx].\decide(\opp)$}{\label{reg-from-con:decide}
                $\assign{\idx}{\idx+1}$\;
                \uIf{$op=\owrite(v)$}{
                    \assign{\val}{v}\;
                }
                \uIf{$\op\neq\bot$}{
                    \uIf{$\opp=\op$}{
                        \lIf{$\op=\oread$}{$\assign{\retval}{\val}$}
                        \lElse{$\assign{\retval}{\ACK}$}
                        {\bf return} $\retval$ {\bf in response to} $\op$\;
                        $\assign{\op}{\bot}$\;
                    }\Else{
                        $C[\idx].\propose(\op)$\;
                    }
                }
            }
        \end{algorithm}
        \vspace*{-5pt}
        \end{minipage}
    }
    \end{tabular}
    \caption{Protocol at a process $p_i$. Let $\mathit{Val}$ be an arbitrary domain of values.
    $\opp$ can either be $\mathit{read}()$ or $\mathit{write}(v)$ for $v\in \mathit{Val}$.}
    \label{fig:reg-from-consensus}
\end{figure}

\begin{lemma}
Let $\FS$ be a fail-prone system and
$\tau: \FS \rightarrow 2^\PP$ be such that
$\forall f=(P,C) \in \FS.\ \tau(f)\neq\emptyset$.
Then the algorithm in Figure~\ref{fig:reg-from-consensus} is an
$(\FS,\tau)$-obstruction-free implementation of a safe register.
\end{lemma}
\begin{proof}
It is easy to see that the algorithm implements a safe register. Thus, we focus on proving
obstruction-freedom.
Let $f \in \FS$, $\sigma$ be an $f$-compliant execution of the algorithm
in Figure~\ref{fig:reg-from-consensus}, $p \in \tau(f)$ and $\opp$ be an operation
invoked by $p$ that executes solo in $\sigma$.
We show that $\opp$ must eventually return.

Let $i$ be such that $\opp$ is first proposed in the $i$-th instance of consensus at a time $t$
and assume by contradiction that $\opp$ never returns.
Since $\opp$ executes solo in $\sigma$, it is not concurrent with any operation invoked by any
other correct process. And because $\opp$ never returns, this means that no operation invocation
by a correct process occurs at any time $\geq t$. Then from time $t$ onwards, the only
outstanding operations are either $\opp$ or operations invoked by faulty processes.
Thus, every proposal by $p$ at a time $\geq t$ executes solo in $\sigma$.

We now show that for each $j \geq i$, $\opp$ is proposed by $p$ in the $j$-th instance of consensus $p.C[j]$.
Because $i$ is defined to be the first instance of consensus where $\opp$ is proposed by $p$, the result
trivially holds for $j=i$. Suppose now that $\opp$ is proposed in the $k$-th instance of consensus
for some $k>i$.
Because the $k$-th instance of consensus is $(\FS,\tau)$-obstruction-free and $\propose(\opp)$ executes
solo in $\sigma$, then $\propose(\opp)$ eventually returns with some value $\opp' \neq \opp$.
Therefore, the handler at line~\ref{reg-from-con:decide} guarantees that $\opp$ is proposed by $p$
in the $(k+1)$-th instance of consensus, as required.

Therefore, there exists $r \geq i$ such that $p$ invokes $\propose(\opp)$ in $C[r]$ after every faulty process has crashed
and thus $p$ is the only process to propose in $C[r]$. Because the $r$-th instance of consensus is
$(\FS,\tau)$-obstruction-free and $\propose(\opp)$ executes solo in $\sigma$,
then $C[r]$ decides some value $\opp' \neq \opp$. Because $p$ is the only
process to propose in $C[r]$ and since the decided value must have been proposed, this contradicts the safety
properties of consensus. The contradiction shows that $\opp$ must eventually return in $\sigma$.
\end{proof}

    \section{Proof of Theorem~\ref{theorem:nxp}}
\label{sec:app-sync}

Consider a fail-prone system $\FS$ satisfying the conditions of
Theorem~\ref{theorem:nxp}: for each $f \in \FS$, the graph $\RG{\GG}{f}$
contains a connected core. For the remainder of the section we fix a failure
pattern $f \in \FS$ and let $\CQ$ be the corresponding connected core.
We now show that every $f$-compliant execution of the algorithm in
Figure~\ref{fig:synchronizer_implementation} over the model $\MCU$ satisfies the
properties in Figure~\ref{fig:synchronizer_specification}.
The proof relies on the following propositions, whose easy proof we omit.

\begin{proposition}
    \label{proposition:tjl}
    If a message $\WISH(V)$ is sent at a time $t$, then $V[i] \leq p_i.\currview(t) + 1$
    for every process $p_i$.
\end{proposition}

\begin{proposition}
    \label{proposition:sjf}
    If a process $p_i$ has $p_i.\views(t)[j] > 0$, then there exists an array $V$
    and a time $t' \leq t$ such that $V[j] = p_i.\views(t)[j]$ and $p_j$ sends
    $\WISH(V)$ at $t'$.
\end{proposition}

\begin{proposition}
    \label{proposition:ivd}
    If a process $p_i$ enters a view $v$, then there exists a process $p_j \in \CQ$,
    an array $V$ and a time $t < E_i(v)$ such that $V[j] \geq v$ and $p_j$ sends
    $\WISH(V)$ at $t$.
\end{proposition}

\begin{proposition}
    \label{proposition:ypq}
    If a process $p_i$ sends $\WISH(V)$ with $V[i] = v+1$ at a time $t$, then
    $A_i(v)\fdef\ \land\ A_i(v) \leq t$.
\end{proposition}

\begin{lemma}[Validity]
    \label{lemma:rcq}
    A process only enters $v + 1$ if some process from $\CQ$ has attempted
        to advance from $v$:
\[
        \forall i, v.\ E_i(v + 1)\fdef\ \implies \AF{\CQ}(v)\fdef\ \land\
        \AF{\CQ}(v) < E_i(v + 1).
\]
\end{lemma}
\begin{proof}
    Let $i$ and $v$ be such that $E_i(v + 1)\fdef$.
    By Proposition \ref{proposition:ivd}, there exists a process $p_j \in \CQ$, an array $V$
    and a time $t < E_i(v + 1)$ such that $V[j] \geq v + 1$ and $p_j$ sends $\WISH(V)$
    at $t$.
    Suppose $V[j]=v+1$.
    Then, by Proposition \ref{proposition:ypq}, $A_j(v)\fdef\ \land\ A_j(v) \leq t$.
    Therefore, $\AF{\CQ}(v)\fdef$ and $\AF{\CQ}(v) \leq A_j(v) \leq t < E_i(v + 1)$.

    Suppose now $V[j] > v + 1$.
    By Proposition \ref{proposition:tjl}, $p_j.\currview(t) \geq V[j] - 1 > v$.
    Let $p_k$ be the first process to enter a view $v_k > v$ at a time $t_k \leq t$.
    The process $p_k$ enters $v_k$ upon executing line \ref{protocol:sync:wish:newview},
    and by lines \ref{protocol:sync:wish:v} and \ref{protocol:sync:wish:guard},
    $p_k.\views(t_k)$ includes more than $\frac{n}{2}$ entries $\geq v_k$.
    Because there are at most $\frac{n}{2}$ processes not in $\CQ$, there exists a process $p_l \in \CQ$
    and a view $v_l \geq v_k$ such that $p_k.\views(t_k)[l] = v_l$.
    By Proposition \ref{proposition:sjf}, there exists an array $V'$ and a time $t_l \leq t_k$ such
    that $V'[l] = v_l$ and $p_l$ sends $\WISH(V')$ at $t_l$.
    By Proposition \ref{proposition:tjl}, $p_l.\currview(t_l) \geq v_l - 1 \geq v_k - 1 \geq v$.
    Because $p_k$ is the first process to enter a view $>v$ at $t_k$ and $t_l \leq t_k$
    then $p_l.\currview(t_l) \leq v$.
    Therefore, $p_l.\currview(t_l) = v$ and $v_l = v + 1$.
    Then, by Proposition \ref{proposition:ypq}, $A_l(v)\fdef\ \land\ A_l(v) \leq t_l$.
    Therefore, $\AF{\CQ}(v)\fdef$ and $\AF{\CQ}(v) \leq A_l(v) \leq t_l \leq t_k \leq t
    < E_i(v + 1)$.
\end{proof}

\begin{lemma}
    \label{lemma:yqh}
    $\CQ$ does not skip views:
    \[
      \forall v, v'.\ 0 < v < v' \land \AEF(v')\fdef\ \implies \EF{\CQ}(v)\fdef\
      \land\ \EF{\CQ}(v) < \AEF(v').
      \]
\end{lemma}
\begin{proof}
    Let $v' \geq 2$ and assume that a process enters $v'$, so that $\AEF(v')\fdef$.
    We prove by induction that for each $k$ satisfying $1 \leq k \leq v' - 1$,
    some process in $\CQ$ enters $v' - k$ earlier than $\AEF(v')$ and thus
    $\EF{\CQ}(v' - k)\fdef\ \land\ \EF{\CQ}(v' - k) < \AEF(v')$.

    For the base case assume that a process enters $v'$.
    Then by Lemma \ref{lemma:rcq}, there exists a process $p_i \in \CQ$ such that
    $A_i(v' - 1)\fdef\ \land\ A_i(v' - 1) < \AEF(v')$.
    Because $p_i.\currview(A_i(v' - 1)) = v' - 1$ then $E_i(v' - 1)\fdef\ \land\
    E_i(v' - 1) \leq A_i(v' - 1) < \AEF(v')$.

    For the inductive step assume that the required holds for some $k$
    so that $\EF{\CQ}(v' - k)\fdef\ \land\ \EF{\CQ}(v' - k) < \AEF(v')$.
    Then by Lemma \ref{lemma:rcq}, there exists a process $p_i \in \CQ$ such that
    $A_i(v' - k - 1)\fdef\ \land\ A_i(v' - k - 1) < \EF{\CQ}(v' - k)$.
    Because $p_i.\currview(A_i(v' - k - 1)) = v' - k - 1$ then
    $E_i(v' - k - 1)\fdef\ \land\ E_i(v' - k - 1) \leq A_i(v' - k - 1) <
    \EF{\CQ}(v' - k) < \AEF(v')$.
\end{proof}

\begin{lemma}
    \label{lemma:lgk}
    Let $\Delta={\tt diameter}(\CQ)\delta$.
    Consider a view $v > 0$ and assume that $v$ is entered by some process in $\CQ$.
    If $\EF{\CQ}(v) \geq \GST$ and no process in $\CQ$ attempts to
    advance from $v$ before $\EF{\CQ}(v) + \Delta$, then all processes in $\CQ$
    enter $v$ and $\EL{\CQ}(v) \leq \EF{\CQ}(v) + \Delta$.
\end{lemma}
\begin{proof}
    Suppose there exists a process $p_i$ and a time $t \leq \EF{\CQ}(v) + \Delta$
    such that $p_i.\currview(t) = v' > v$.
    By Lemma \ref{lemma:yqh}, $\AEF(v + 1)\fdef\ \land\ \AEF(v + 1) \leq \AEF(v')
    \leq t$.
    Thus, by Lemma \ref{lemma:rcq}, $\AF{\CQ}(v)\fdef\ \land\
    \AF{\CQ}(v) < \AEF(v + 1) \leq t \leq \EF{\CQ}(v) + \Delta$, contradicting
    the assumption that no process in $\CQ$ attempts to advance from
    $v$ before $\EF{\CQ}(v) + \Delta$.
    Therefore, for all times $t \leq \EF{\CQ}(v) + \Delta$ and processes $p_i$,
    $p_i.\currview(t) \leq v$.

    Let $p_i$ be the process in $\CQ$ to enter $v$ at the time $\EF{\CQ}(v)$
    and $p_j \in \CQ$.
    We prove by induction on the distance $d$ from $p_i$ to $p_j$ (considering the
    reliable links that make up $\CQ$) that $E_j(v)\fdef \land E_j(v) \leq E_i(v) + d\delta$.
    Since $p_i$ is the only process at distance $0$ from $p_i$, the result trivially
    holds for $d = 0$.
    Assume now the result holds for all processes in $\CQ$ at a distance $d$ from $p_i$
    and suppose the distance from $p_i$ to $p_j$ is $d + 1$.
    Therefore, there exists a process $p_k \in \CQ$ such that the distance from $p_i$ to $p_k$
    is $d$ and the distance from $p_k$ to $p_j$ is $1$.
    By induction hypothesis $p_k$ enters $v$ no later than $E_i(v) + d\delta$.
    The process $p_k$ enters view $v$ upon executing line \ref{protocol:sync:wish:newview},
    and by lines \ref{protocol:sync:wish:v} and \ref{protocol:sync:wish:guard}, $p_k.\views(E_k(v))$ includes
    more than $\frac{n}{2}$ entries $\geq v$.

    Line \ref{protocol:sync:wish:wish} guarantees that upon entering $v$, $p_k$ sends a
    $\WISH(p_k.\views(E_k(v)))$  to every process.
    Since the link between $p_k$ and $p_j$ is timely after $\GST$, the
    $\WISH$ message is received by $p_j$ at a time $t_j \leq E_i(v) + (d + 1)\delta$.
    Upon receipt of the $\WISH$ message, $p_j$ executes lines
    \ref{protocol:sync:wish:viewsfor} and \ref{protocol:sync:wish:views},
    ensuring that $p_j.\views(t_j) \geq p_k.\views(E_k(v))$.
    Therefore, $p_j.\views(t_j)$ includes more than $\frac{n}{2}$ entries $\geq v$, and thus $p_j$
    is guaranteed to enter a view $v' \geq v$ no later than $t_j$.
    Because $t_j \leq \EF{\CQ}(v) + \Delta$ and no process can have
    a view $> v$ before $\EF{\CQ}(v) + \Delta$, it follows that $v' = v$.
    Therefore, $p_j$ enters $v$ no later than $t_j$ and thus
    $E_j(v) \leq t_j \leq E_i(v) + (d + 1)\delta$.
    Thus, all processes in $\CQ$ enter $v$ and $\EL{\CQ}(v) \leq \EF{\CQ}(v) + \Delta$.
\end{proof}

\begin{lemma}[Bounded Entry]
    \label{lemma:awl}
    For some $\V$ and $d$, if a process from $\CQ$ enters a view
    $v \geq \V$ and no process from $\CQ$ attempts to advance to a higher view within time $d$,
    then every process from $\CQ$ will enter $v$ within $d$:
    \begin{align*}
      &\exists \V,d.\ \forall v \geq \V.\ \EF{\CQ}(v)\fdef\ \land\
      \neg(\AF{\CQ}(v) < \EF{\CQ}(v) + d) \implies \\
      & (\forall p_i \in \CQ.\ E_i(v)\fdef) \land (\EL{\CQ}(v) \leq \EF{\CQ}(v) + d).
      \end{align*}
\end{lemma}
\begin{proof}
    Let $\V = \max\{v\ |\ \EF{\CQ}(v)\fdef\ \land\ \EF{\CQ}(v) < \GST\} + 1$
    so that $\forall v \geq \V.\ \EF{\CQ}(v)\fdef\ \implies \EF{\CQ}(v) \geq \GST$
    and $d = {\tt diameter}(\CQ)\delta$.
    Then, by Lemma \ref{lemma:lgk}, Bounded Entry holds.
\end{proof}

Given arrays $V$ and $V'$ we say that $V \leq V'$ iff $\forall i.\ V[i] \leq V'[i]$.

\begin{lemma}
    \label{lemma:nvg}
    If a process sends $\WISH(V)$ before sending $\WISH(V')$, then $V \leq V'$.
\end{lemma}
\begin{proof}
    Let $p_j \in \mathcal{P}$ and $t$ and $t'$ be the times at which a process $p_i$ sends $\WISH(V)$
    and $\WISH(V')$, respectively.
    Notice that $V[j] = p_i.\views(t)[j]$ and $V'[j] = p_i.\views(t')[j]$.
    Since $p_i.\views[j]$ is non-decreasing, as guaranteed by lines
    \ref{protocol:sync:adv:views} and \ref{protocol:sync:wish:views},
    $p_i.\view(t)[j] \leq p_i.\view(t')[j]$.
    Therefore, $V[j] \leq V'[j]$. Since $p_j$ was picked arbitrarily, $V \leq V'$.
\end{proof}

\begin{lemma}
    \label{lemma:pzt}
    For all processes $p_i \in \CQ$, times $t \geq \OGST$ and arrays
    $V$, if $p_i$ sends $\WISH(V)$ at a time $\leq t$, then there exists
    an array $V' \geq V$ and a time $t'$ such that $\GST \leq t' \leq t$
    and $p_i$ sends $\WISH(V')$ at $t'$.
\end{lemma}
\begin{proof}
    Let $s \leq t$ be the time at which $p_i$ sends $\WISH(V)$.
    If $s \geq \GST$, then choosing $t' = s$ and $V' = V$ validates the lemma.
    Assume now that $s < \GST$.
    Since after $\GST$ the $p_i$'s local clock advances at the same rate as real time,
    there exists a time $t'$ satisfying $\GST \leq t' \leq t$ such that $p_i$
    executes the retransmission code in lines \ref{protocol:sync:periodically}
    and \ref{protocol:sync:periodically:wish} at $t'$.
    Then, there exists an array $V'$ such that $p_i$ sends $\WISH(V')$ at $t'$.
    By Lemma \ref{lemma:nvg}, $V' \geq V$.
\end{proof}

\begin{lemma}[Startup]
    \label{lemma:rlw}
    If more than $\frac{n}{2}$ processes from $\CQ$ invoke $\adv$, then some process from $\CQ$
    will enter view $1$:
    \[
      (\exists P \subseteq \CQ.\ |P| > \frac{n}{2}\ \land\ (\forall p_i \in P.\
      A_i(0)\fdef)) \implies \EF{\CQ}(1)\fdef.
      \]
\end{lemma}
\begin{proof}
    Assume by contradiction that there exists a set $P \subseteq \CQ$ of more than $\frac{n}{2}$
    processes such that $\forall p_i \in P.\ A_i(0)\fdef$, and
    no process in $\CQ$ enters view $1$. By Lemma \ref{lemma:yqh}, the latter implies
    \begin{equation}\label{lemma:rlw:eq1}
        \forall v > 0.\ \forall p_i \in \CQ.\ E_i(v)\fndef.
    \end{equation}
    We now show that for every $p_i \in \CQ$ and $p_j \in P$, there exists an array $V_j$ and a time $t_j$ such
    that $V_j[j] \geq 1$ and $p_i$ receives $\WISH(V_j)$ at $t_j$.
    Fix $p_j \in P$. We prove the result by induction on the distance $d$ from $p_j$ to $p_i$
    (considering the reliable links that make up $\CQ$).

    Since $A_j(0)\fdef$, lines \ref{protocol:sync:adv:views} and \ref{protocol:sync:adv:wish}
    guarantee that there exists an array $V_j$ and a time $t_j$ such that $V_j[j] = 1$ and
    $p_j$ sends $\WISH(V_j)$ at $t_j$. Since messages sent to itself are always delivered
    and $p_j$ is the only process at distance $0$ from $p_j$, the result trivially holds for $d = 0$.

    Assume now the result holds for all processes at a distance $d$ from $p_j$ and suppose that
    $p_i \in \CQ$ is a process at a distance $d + 1$ from $p_j$.
    Therefore, there exists a process $p_k \in \CQ$ such that the distance from $p_j$ to $p_k$
    is $d$ and the distance from $p_k$ to $p_i$ is $1$.
    By induction hypothesis there exists an array $V_j$ and a time $t_j$ such that
    $V_j[j] \geq 1$ and $p_k$ receives $\WISH(V_j)$ at $t_j$.
    Upon receipt of the $\WISH$ message, $p_k$ executes lines
    \ref{protocol:sync:wish:viewsfor} and \ref{protocol:sync:wish:views},
    ensuring that $p_k.\views(t_j)[j] \geq V[j] \geq 1$.
    Because the periodic handler in line \ref{protocol:sync:periodically}
    fires every $\rho$ units of time, there exists an array $V_{k'}$ and a time $t_{k'} \leq t_j + \rho$
    such that $V_{k'}[j] \geq 1$ and $p_k$ sends $\WISH(V_{k'})$ at $t_{k'}$.
    Let $T_k=\max(\OGST,t_k')$. By Lemma \ref{lemma:pzt},
    there exists an array $V_{k} \geq V_{k'}$ and a time $t_{k}$ such that
    $\GST \leq t_{k} \leq T_{k}$ and $p_k$ sends $\WISH(V_{k})$ at $t_{k}$.
    Since the link between $p_k$ and $p_i$ is
    timely after $\GST$, the $\WISH(V_{k})$ is received by $p_i$ no
    later than $t_{k} + \delta$. And because $V_k \geq V_{k'}$ then $V_k[j] \geq V_{k'}[j] \geq 1$.
    This completes the induction.

    Let $p_l \in \CQ$ be any process and, for each $p_i \in P$, let $V_i$ and $t_i$ be such that
    $V_i[i] \geq 1$ and $p_l$ receives $\WISH(V_i)$ at $t_i$. Let $T=\max t_i$.
    Lines \ref{protocol:sync:wish:viewsfor} and \ref{protocol:sync:wish:views} guarantee
    that $p_l.\views(T)[i] \geq 1$ for each $p_i \in P$. Since $|P| > \frac{n}{2}$,
    $p_k.\views(T)$ includes more than $\frac{n}{2}$ entries $\geq 1$, and thus, $p_l.v'(T) \geq 1$.
    By (\ref{lemma:rlw:eq1}), $p_l.\currview(T) = 0$. Hence, line
    \ref{protocol:sync:wish:guard} ensures that $p_l$ enters a view $\geq 1$ by $T$,
    contradicting (\ref{lemma:rlw:eq1}).
\end{proof}

\begin{lemma}[Progress]
    \label{lemma:msz}
    If a process from $\CQ$ enters a view $v$ and, for some set $P \subseteq \CQ$
    of more than $\frac{n}{2}$ processes, any process in $P$ that enters $v$ eventually invokes $\adv$, then some process from $\CQ$ will enter $v+1$:
    \[
      \forall v.\, \EF{\CQ}(v)\fdef\, \land\, (\exists P \subseteq \CQ.\ |P| >
      \frac{n}{2}\, \land\, (\forall p_i \in P. E_i(v)\fdef \implies A_i(v)\fdef))
      \implies \EF{\CQ}(v + 1)\fdef.
      \]
\end{lemma}
\begin{proof}
    Assume by contradiction that there exists a set $P \subseteq \CQ$ of more than $\frac{n}{2}$ processes
    such that $\forall p_i \in P.\ E_i(v)\fdef\ \implies A_i(v)\fdef$, $\EF{\CQ}(v)\fdef$,
    and no process in $\CQ$ enters view $v + 1$.
    By Lemma \ref{lemma:yqh}, the latter implies
    \begin{equation}\label{lemma:msz:eq1}
        \forall v' > v.\ \forall p_i \in \CQ.\ E_i(v')\fndef.
    \end{equation}
    Let $p_i \in \CQ$ be the first process to enter $v$.
    We show that every process $p_j \in \CQ$ enters $v$ by induction on the
    distance $d$ from $p_i$ to $p_j$ (considering the reliable links that make up $\CQ$).
    Since $p_i$ is the only process at distance $0$ from $p_i$, the result trivially
    holds for $d = 0$.
    Assume now the result holds for all processes in $\CQ$ at a distance $d$ from $p_i$
    and suppose the distance from $p_i$ to $p_j$ is $d + 1$.
    Therefore, there exists a process $p_k \in \CQ$ such that the distance from $p_i$
    to $p_k$ is $d$ and the distance from $p_k$ to $p_j$ is $1$.
    By induction hypothesis $p_k$ enters $v$.
    The process $p_k$ enters view $v$ upon executing line \ref{protocol:sync:wish:newview},
    and by lines \ref{protocol:sync:wish:v} and \ref{protocol:sync:wish:guard}, $p_k.\views(E_k(v))$ includes more than $\frac{n}{2}$ entries $\geq v$.

    Line \ref{protocol:sync:wish:wish} guarantees that upon entering $v$, $p_k$ sends
    $\WISH(p_k.\views(E_k(v)))$ to every process.
    Let $T_k = \max(\OGST, E_k(v))$.
    By Lemma \ref{lemma:pzt},
    there exists an array $V_{k} \geq p_k.\views(E_k(v))$ and a time $t_{k}$ such that
    $\GST \leq t_{k} \leq T_{k}$ and $p_k$ sends $\WISH(V_{k})$ at $t_{k}$.
    Since the link between $p_k$ and $p_j$ is timely after $\GST$, the $\WISH$ message
    is received by $p_j$ at a time $t_j$.
    Upon receipt of the $\WISH$ message, $p_j$ executes lines
    \ref{protocol:sync:wish:viewsfor} and \ref{protocol:sync:wish:views},
    ensuring that $p_j.\views(t_j) \geq V_{k} \geq p_k.\views(E_k(v))$.
    Therefore, $p_j.\views(t_j)$ includes more than $\frac{n}{2}$ entries $\geq v$, and thus
    $p_j$ is guaranteed to enter a view $v' \geq v$ no later than $t_j$.
    By (\ref{lemma:msz:eq1}), $p_j$ never enters a view $> v$ and thus
    $p_j$ enters $v$. This completes the induction.
    Because every process from $\CQ$ enters $v$, we have $\forall p_i \in P.\ E_i(v)\fdef$
    and thus $\forall p_i \in P.\ A_i(v)\fdef$.

    We now show that for every $p_i \in \CQ$ and $p_j \in P$, there exists an array $V_j$
    and a time $t_j$ such that $V_j[j] \geq v + 1$ and $p_i$ receives $\WISH(V_j)$ at $t_j$.
    Fix $p_j \in P$. We prove the result by induction on the distance $d$ from $p_j$ to $p_i$
    (considering the reliable links that make up $\CQ$).

    Since $A_j(v)\fdef$, lines \ref{protocol:sync:adv:views} and \ref{protocol:sync:adv:wish}
    guarantee that there exists an array $V_j$ and a time $t_j$ such that $V_j[j] = v + 1$
    and $p_j$ sends $\WISH(V_j)$ at $t_j$.
    Since messages sent to itself are always delivered and $p_j$ is the only process at
    distance $0$ from $p_j$, the result trivially holds for $d = 0$.

    Assume now the result holds for all processes at a distance $d$ from $p_j$ and suppose
    that $p_i \in \CQ$ is a process at a distance $d + 1$ from $p_j$.
    Therefore, there exists a process $p_k \in \CQ$ such that the distance from $p_j$ to $p_k$
    is $d$ and the distance from $p_k$ to $p_i$ is $1$.
    By induction hypothesis there exists an array $V_j$ and a time $t_j$ such that
    $V_j[j] \geq v + 1$ and $p_k$ receives $\WISH(V_j)$ at $t_j$.
    Upon receipt of the $\WISH$ message, $p_k$ executes lines
    \ref{protocol:sync:wish:viewsfor} and \ref{protocol:sync:wish:views},
    ensuring that $p_k.\views(t_j)[j] \geq V[j] \geq v + 1$.
    Because the periodic handler in line \ref{protocol:sync:periodically}
    fires every $\rho$ units of time, there exists an array $V_{k'}$ and a time $t_{k'} \leq t_j + \rho$
    such that $V_{k'}[j] \geq v + 1$ and $p_k$ sends $\WISH(V_{k'})$ at $t_{k'}$.
    Let $T_k=\max(\OGST,t_k')$. By Lemma \ref{lemma:pzt},
    there exists an array $V_{k} \geq V_{k'}$ and a time $t_{k}$ such that
    $\GST \leq t_{k} \leq T_{k}$ and $p_k$ sends $\WISH(V_{k})$ at $t_{k}$.
    Since the link between $p_k$ and $p_i$ is
    timely after $\GST$, the $\WISH(V_{k})$ is received by $p_i$ no
    later than $t_{k} + \delta$. And because $V_k \geq V_{k'}$ then $V_k[j] \geq V_{k'}[j] \geq v + 1$.
    This completes the induction.

    Let $p_l \in \CQ$ be any process and, for each $p_i \in P$, let $V_i$ and $t_i$ be such that
    $V_i[i] \geq v + 1$ and $p_l$ receives $\WISH(V_i)$ at $t_i$. Let $T=\max t_i$.
    Lines \ref{protocol:sync:wish:viewsfor} and \ref{protocol:sync:wish:views} guarantee
    that $p_l.\views(T)[i] \geq v + 1$ for each $p_i \in P$. Since $|P| > \frac{n}{2}$,
    $p_k.\views(T)$ includes more than $\frac{n}{2}$ entries $\geq v + 1$, and thus, $p_l.v'(T) \geq v + 1$.
    By (\ref{lemma:msz:eq1}), $p_l.\currview(T) \leq v$. Hence, line
    \ref{protocol:sync:wish:guard} ensures that $p_l$ enters a view $\geq v + 1$ by $T$,
    contradicting (\ref{lemma:msz:eq1}).
\end{proof}

\begin{proof}[Proof of Theorem~\ref{theorem:nxp}]
    The guard at line \ref{protocol:sync:wish:guard} ensures Monotonicity.
    Validity, Bounded Entry, Startup and Progress are given by Lemmas
    \ref{lemma:rcq}, \ref{lemma:awl}, \ref{lemma:rlw} and \ref{lemma:msz}, respectively.
\end{proof}

    \section{Proof of Theorem~\ref{thm:upper-consensus}}
\label{sec:app-consensus}

Fix a failure pattern $f$ and let $\CQ$ be the corresponding connected core
guaranteed to exist by the assumptions of Theorem~\ref{thm:upper-consensus}.

\begin{proof}[Proof of Lemma~\ref{lemma:not_ok_adv}]
	Upon entering $v$ the process $p_i$ starts $\timerdecision$ (line
	\ref{protocol:consensus:newview:starttimer}).
        The process $p_i$ cannot stop the timer at line
        \ref{protocol:consensus:newview:starttimer}, since together with
        Monotonicity, this would imply that $p_i$ enters a view higher than $v$.
	The process $p_i$ cannot stop the timer at
	line \ref{protocol:consensus:decided:stoptimer} either, since this would imply
        that $p_i$ decides in $v$.
	Therefore the timer must expire, after which $p_i$ invokes $\adv$ in
        $v$ (line \ref{protocol:consensus:timerexpires:adv}).
\end{proof}

Lemma~\ref{lemma:ok_stay_all} together with Validity implies that all processes in
$\CQ$ will stay in a view where progress is possible provided the timers are
high enough. However, to apply it we need to argue that, if processes from $\CQ$
keep changing views due to lack of progress, all of them will increase their
timeouts high enough to satisfy the bounds in Lemma~\ref{lemma:ok_stay_all}. To
this end, we prove the following generalization of
Lemma~\ref{lemma:ok_stay_all}: in a sufficiently high view $v$ with a leader
from the connected core, if the timeout at a process from $\CQ$ that enters $v$
is high enough, then this process cannot be the first to initiate a view
change. Hence, for the protocol to enter another view, some other process with a
lower timeout must call $\adv$ and thus increase its duration.
Let $\Delta = (\delta + \rho) \diameter(\CQ)$.

\begin{lemma}
	\label{lemma:ok_stay}
	Let $v \geq \V$ be a view such that $\leader(v) \in \CQ$,
        $\EF{\CQ}(v) \geq \GST$ and $\leader(v)$ invokes $\propose$ no later
        than $\EF{\CQ}(v)$.  If $p_i \in \CQ$ enters $v$ and
        $\durdecision_i(v) > d + 3\Delta$, then $p_i$ is not the first process
        in $\CQ$ to invoke $\adv$ in $v$.
\end{lemma}

The proof of Lemma~\ref{lemma:ok_stay} relies on the following technical proposition.
It states that the processes from
$\CQ$ exchange their most up to date information within time $\Delta$, as long as they
do not enter a higher view.

\begin{proposition}
	\label{proposition:fwd}
			Let $t \geq \GST$ be such that no process in $\CQ$ enters a view higher than $v$ before
      $t + \Delta$. Then
      \begin{enumerate}
    	  \item $\forall p_i, p_j, p_k \in \CQ.\ \ \ p_i.\MSGOB(t)[k].\itview=v {\implies}
      		p_j.\MSGOB (t + \Delta)[k] = p_i.\MSGOB (t)[k]$.
      	\item $\forall p_i, p_j \in \CQ.\ \forall p_k.\  p_i.\MSGTA(t)[k].\itview=v {\implies}
      		p_j.\MSGTA (t + \Delta)[k] = p_i.\MSGTA (t)[k]$.
      	\item $\forall p_i, p_j \in \CQ.\ \forall p_k.\  p_i.\MSGTB(t)[k].\itview=v {\implies}
      		p_j.\MSGTB (t + \Delta)[k] = p_i.\MSGTB(t)[k]$.
    	\end{enumerate}
\end{proposition}

Informally, each process propagates its state every $\rho$ units of time
(line \ref{protocol:consensus:periodically:state}).
Upon receipt of the state (line \ref{protocol:consensus:state}), processes update
their vectors with the most up to date information (lines \ref{protocol:consensus:state:msgob},
\ref{protocol:consensus:state:msgta} and \ref{protocol:consensus:state:msgtb}).
Because no process from $\CQ$ is further than $\diameter(\CQ)$ from $p_i$, its state is guaranteed
to be received by every other process in $\CQ$ no later than $t + \Delta$.
And because no process from $\CQ$ enters a view higher than $v$ at least until after $t + \Delta$,
these values remain unchanged until then.

\begin{proof}[Proof of Proposition \ref{proposition:fwd}]
	Let's consider the case of $\MSGTB$; the other cases are similar.
	Fix $p_i \in \CQ$ and let $t \geq \GST$ be a time, $p_k \in \PP$ and $x$ be a value such that
	$p_i.\MSGTB(t)[k]=(v,x)$. Let $p_j \in \CQ$ and $\mathsf{dist}(p_i,p_j)$ be
	the distance from $p_i$ to $p_j$ in $\RG{\GG}{f}$. We now show that
	$p_j.\MSGTB(t + (\rho + \delta)\mathsf{dist}(p_i,p_j))[k] = (v,x)$.
	We prove the result by induction on the distance $d$ from $p_i$ to $p_j$.
	Since $p_i$ is the only process at distance $0$ from $p_i$, the result
	trivially holds for $d=0$. Assume now that the result holds for all processes in $\CQ$
	at distance $d$ from $p_i$ and suppose that the distance from $p_i$ to $p_j$ is $d+1$.
	Therefore, there exists a process $p_m \in \CQ$ such that the distance from $p_i$
	to $p_m$ is $d$ and the distance from $p_m$ to $p_j$ is $1$. By induction hypothesis,
	$p_m.\MSGTB(t + (\rho + \delta)d)[k] = (v,x)$.
	Because no process in $\CQ$ enters a view higher than $v$ before $t + \Delta$,
	the value of $p_m.\MSGTB(t + (\rho + \delta)d)[k]$ remains unchanged from
	$t + (\rho + \delta)d$ until $t + \Delta$.

	Because the handler at line \ref{protocol:consensus:periodically} runs every $\rho$
	units of time, there exists a time $t_s$ and a vector $\VTB$ such that
	$t + (\rho + \delta)d \leq t_s \leq t + (\rho + \delta)d + \rho$,
	$\VTB[k]=(v,x)$ and $p_m$ sends $\STATE(\_,\_,\VTB)$ at $t_s$
	(line \ref{protocol:consensus:periodically:state}).
	Since the channel between $p_m$ and $p_j$ is reliable after $\GST$,
	the $\STATE(\_,\_,\VTB)$ message is received by $p_j$ at a time $t_r$ such that
	$t_r \leq t_s+\delta \leq t + (\rho + \delta)d + \rho + \delta = t + (\rho + \delta)(d+1)$.
	Thus the process $p_j$ executes the handler at line \ref{protocol:consensus:state}
	by $t_r$.

	If $p_j.\MSGTB(t_r)[k].\itview = v$, and since at most one value can be proposed
	per view, then $p_j.\MSGTB(t_r)[k]=(v,x)$. If $p_j.\MSGTB(t_r)[k].\itview < v$,
	then $p_j$ sets $p_j.\MSGTB(t_r)[k]=(v,x)$ (line \ref{protocol:consensus:state:msgtb}).
	Because no process in $\CQ$ enters a view higher than $v$ before $t + \Delta$,
	then $p_j.\MSGTB(t_r)[k].\itview$ cannot be greater than $v$, and also $p_j.\MSGTB$ remains
	unchanged from $t_r$ until $t + \Delta$. Therefore,
	$p_j.\MSGTB(t + (\rho + \delta)(d+1))[k]=(v,x)$.
\end{proof}

\begin{proof}[Proof of Lemma \ref{lemma:ok_stay}]
	Let $T = \EF{\CQ}(v) + d + 3\Delta$.
	By contradiction, assume that $p_i$ is the first process in $\CQ$ to call $\adv$ in $v$.
	This only occurs if $\timerdecision$ expires at $p_i$ (line \ref{protocol:consensus:timerexpires:adv}).
	Because $p_i$ is the first process to call $\adv$ in $v$ and $\durdecision_i(v) >
	d + 3\Delta$, no process in $\CQ$ invokes $\adv$ in $v$ until after $T$.
	Then by Bounded Entry all processes in $\CQ$ enter $v$ by $\EL{\CQ}(v) \le \EF{\CQ}(v) + d$.
	By Validity no process can enter $v + 1$ until after $T$,
	and by Lemma \ref{lemma:yqh}, the same holds for any view $> v$.
	Thus, all processes in $\CQ$ stay in $v$ at least until $T$.
	We now show that during this time processes in $\CQ$
	exchange the messages needed for $p_i$ to decide and stop the timer,
        thereby reaching a contradiction.

	Let $p_l = \leader(v)$. We first prove that $p_l$ proposes a value in
        view $v$.
	Indeed, pick an arbitrary process $p_j \in \CQ$. Upon entering $v$, $p_j$ sets
	$\MSGOB[j]=(v,\_,\_)$ and $\phase = \ENTERED$
	(lines \ref{protocol:consensus:newview:msg1b} and \ref{protocol:consensus:newview:phase}).
	By Proposition \ref{proposition:fwd}, $p_l.\MSGOB(E_j(v) + \Delta)[j] = p_j.\MSGOB(E_j(v))[j]$.
	Because no process enters a view $> v$ until after $T$, the value of $p_l.\MSGOB[j]$
	remains unchanged from $E_j(v) + \Delta$ until $T$.
	Then, since $p_j$ was picked arbitrarily, by $\EL{\CQ}(v) + \Delta$ the
        array $p_l.\MSGOB$ contains $|\CQ| > \frac{n}{2}$ entries with view $v$.
	Thus, the precondition at line \ref{protocol:consensus:proposed} is satisfied at $p_l$
	at a time $t_{pl}$ such that $\EF{\CQ}(v) \leq t_{pl} \leq \EL{\CQ}(v) + \Delta \leq \EF{\CQ}(v) + d + \Delta$.
	Since $p_l$ invokes $\propose$ no later than $\EF{\CQ}(v)$, $p_l.\pval(t_{pl}) \neq \bot$.
	Therefore, by $t_{pl}$ the process $p_l$ sets $\MSGTA[l]=(v,x)$ for some $x$,
	and $\phase = \PROPOSED$ (lines \ref{protocol:consensus:msgta1},
	\ref{protocol:consensus:msgta2} and \ref{protocol:consensus:phase}).
	Moreover, the guard at line \ref{protocol:consensus:proposed} and line \ref{protocol:consensus:phase}
	guarantee that the proposed value $x$ is unique for view $v$.

	We next prove that all processes in $\CQ$ accept the proposal $x$ made by
        $p_l$ in view $v$. Indeed, pick an arbitrary process $p_j \in \CQ$.
	By Proposition \ref{proposition:fwd}, $p_j.\MSGTA(t_{pl} + \Delta)[l] = p_l.\MSGTA(t_{pl})[l]$.
	Thus, the precondition at line \ref{protocol:consensus:accepted} is satisfied at
	$p_j$ at a time $t_{aj}$ such that $t_{aj} \leq t_{pl} + \Delta \leq \EF{\CQ}(v) + d + 2\Delta$.
	Therefore, by $t_{aj}$ the process $p_j$ sets
	$\MSGTB[j] = (v,x)$ and $\phase = \ACCEPTED$ (lines \ref{protocol:consensus:accepted:msgtb} and
	\ref{protocol:consensus:accepted:phase}).

	Finally, we prove that $p_i$ decides in view $v$.
	Indeed, pick an arbitrary process $p_j \in \CQ$.
	By Proposition \ref{proposition:fwd}, $p_i.\MSGTB(t_{aj} + \Delta)[j] = p_j.\MSGTB(t_{aj})[j] = (v,x)$.
	Because no process enters a view $> v$ until after $T$, the value of $p_i.\MSGTB[j]$
	remains unchanged from $t_{aj} + \Delta$ until $T$.
	Then, since $p_j$ was picked arbitrarily, by $\max\{t_{aj} \mid p_j\in\CQ\} + \Delta$
	the array $p_i.\MSGTB$ contains $|\CQ| > \frac{n}{2}$ entries equal to $(v,x)$.
	Thus, the precondition at line \ref{protocol:consensus:decided} is satisfied at $p_i$
	at a time $t_{di}$ such that $t_{di} \leq \EF{\CQ}(v) + d + 3\Delta = T$.
	Therefore, by $t_{di}$ the process $p_i$ stops the timer $\timerdecision$
	(line \ref{protocol:consensus:decided:stoptimer}) before it expires,
	which contradicts the fact that $\timerdecision$ expires at $p_i$ while in $v$.
\end{proof}

Finally, we prove the liveness of the algorithm in
Figure~\ref{fig:consensus_protocol}, from which
Theorem~\ref{thm:upper-consensus} follows.
\begin{theorem}
  Let $\FS$ be a fail-prone system such that for all $f \in \FS$, the graph $\RG{\GG}{f}$ contains a
  connected core $\CQ_f$, and let $\tau: \FS \rightarrow 2^\PP$ be such that
  $\forall f \in \FS.\ \tau(f) = \CQ_f$. Then the algorithm in
  Figure~\ref{fig:consensus_protocol} is $(\FS,\tau)$-wait-free over the model
  $\MCU$ (partially synchronous / eventually reliable / flaky).
\end{theorem}
\begin{proof}
	By contradiction, assume there exists $f \in \FS$, $p_j \in \tau(f)=\CQ$ and an $f$-compliant
	fair execution of the algorithm in Figure~\ref{fig:consensus_protocol} such that
	$p_j$ invokes $\propose$ at a time $t_{pj}$ but that the operation never returns.
		We first prove that in this case the protocol keeps moving through views
	forever.

  	\setcounter{myclaim}{0}

	\begin{myclaim}
	\label{clm:all_views_entered}
	Every view is entered by some process in $\CQ$.
	\end{myclaim}
	\begin{proof}
	Since all processes in $\CQ$ invoke $\adv$ upon starting, by Startup some
	process in $\CQ$ enters view $1$.
	We now show that infinitely many views are entered by some process in $\CQ$.
	Assume the contrary, so that there exists a maximal view $v$ entered by any process in $\CQ$.
	Let $P \subseteq \CQ$ be any set of more than $\frac{n}{2}$ processes and
	consider an arbitrary process $p_i \in P$ that enters $v$.

	Suppose $p_i$ decides in $v$.
	Then there exists a time $t$, a view $v' \geq v$ and a value $x$ such that
	$p_i.\MSGTB(t)$ contains
	more than $\frac{n}{2}$ entries equal to $(v',x)$.
        Let $Q$ be the set of processes these entries correspond to. Then one of them
        must be from $\CQ$.
	Because no process from $\CQ$ enters a view $> v$, the precondition at line
	\ref{protocol:consensus:accepted} ensures that $v' = v$.
	For similar reasons, the value of $p_i.\MSGTB$ remains unchanged from $t$ onwards.
	Then there exists a time $t' \geq \GST$ such that $p_i.\MSGTB(t')$ contains
	more than $\frac{n}{2}$ entries equal to $(v,x)$.

	By Proposition \ref{proposition:fwd}, $p_j.\MSGTB(t' + \Delta)[k] = p_i.\MSGTB(t')[k]$
	for each $p_k \in Q$.
	Because no process from $\CQ$ enters a view $> v$,
	the value of $p_j.\MSGTB$ remains unchanged from $t' + \Delta$ onwards.
	Then the precondition at line \ref{protocol:consensus:decided}
	is satisfied at $p_j$ at a time $t_{dj}$. Therefore, by $t_{dj}$
	the process $p_j$ sets $\val = x$ and $\phase = \DECIDED$.
	But then the $\propose$ operation invoked by $p_j$ returns by
	$\max\{t_{pj},t_{dj}\}$, contradicting the assumption that it never does so.

	Therefore, $p_i$ never decides in $v$.
	Then by Lemma \ref{lemma:not_ok_adv}, $p_i$ eventually invokes $\adv$ in $v$.
	Since $p_i$ was picked arbitrarily, we have $\forall p_i \in P.\ E_i(v)\fdef
	\implies A_i(v)\fdef$.
	By Progress, $\EF{\CQ}(v + 1)\fdef$, which contradicts the fact that $v$
	is the maximal view entered by any process in $\CQ$.
	Thus, processes in $\CQ$ keep entering views forever.
	The claim then follows from Lemma \ref{lemma:yqh} ensuring that, if a view is entered by a process
	in $\CQ$, then so are all preceding views.
	\end{proof}

	\begin{myclaim}
	\label{clm:exp_inf_times}
	Every process in $\CQ$ executes the timer expiration handler at line \ref{protocol:consensus:timerexpires}
	infinitely often.
	\end{myclaim}
	\begin{proof}
	Assume the contrary and let $\CF$ and $\CI$ be the sets of
	processes in $\CQ$ that execute the timer expiration handler finitely and
	infinitely often, respectively. Then $\CF \neq \varnothing$, and by Claim
	\ref{clm:all_views_entered} and Validity, $\CI \neq \varnothing$.
		Let view $v_2$ be the first view such that $v_2 \geq \V$ and
        $\EF{\CQ}(v_2)$ $\geq \max\{\GST,t_{pj}\}$;
	such a view exists by Claim \ref{clm:all_views_entered}.
	The $\durdecision$ value increases unboundedly at processes from $\CI$, and does not change
	after some view $v_3$ at processes from $\CF$.
	By Claim \ref{clm:all_views_entered} and since leaders rotate round-robin, there exists a view
	$v_4 \geq \max\{v_2, v_3\}$ led by $p_j$ such that any process
	$p_i \in \CI$ that enters $v_4$ has $\durdecision_i(v_4) > d + 3\Delta$.
	By Claim \ref{clm:all_views_entered} and Validity, at least one process in $\CQ$ calls $\adv$
	in $v_4$; let $p_k$ be the first process to do so. Because $v_4 \geq v_3$,
	this process cannot be in $\CF$, since none of these processes can
	increase their timers in $v_4$. Then $p_k \in \CI$, contradicting
	Lemma~\ref{lemma:ok_stay}.
	\end{proof}

Since a process increases $\durdecision$ every time $\timerdecision$
expires, by Claim~\ref{myclaim:exp_inf_times} all processes will eventually 
have $\durdecision > d + 3 (\delta + \rho) \diameter(\CQ)$. Since leaders rotate
round-robin, by Claim 1 there will be infinitely many views led by $p_j$.
Hence, there exists a view $v_1 \geq \V$ led by $p_j$ such that
$\EF{\CQ}(v_1) \geq \max\{\GST,t_{pj}\}$ and for any process $p_i \in \CQ$ that
enters $v_1$ we have
$\durdecision_i(v_1) > d + 3 (\delta + \rho) \diameter(\CQ)$.  Then by
Lemma~\ref{lemma:ok_stay_all}, no process in $\CQ$ calls $\adv$ in $v_1$. On the
other hand, by Claim 1 some process in $\CQ$ enters $v_1 + 1$.  Then by
Validity, some process in $\CQ$ calls $\adv$ in $v_1$, which is a contradiction.
\end{proof}

\fi

\end{document}